\documentclass[lettersize,journal]{IEEEtran}
\usepackage{amsmath,amsfonts}
\usepackage{array}
\usepackage{threeparttable}
\usepackage{tikz}
\usepackage{textcomp}
\usepackage{stfloats}
\usepackage{url}
\usepackage{verbatim}
\usepackage{graphicx}
\usepackage{cite}
\usepackage{listings}
\usepackage{lipsum}
\usepackage{caption}
\usepackage{graphicx}
\usepackage{subcaption}
\usepackage{amsthm}
\usepackage{algorithmicx}
\usepackage{algcompatible}
\usepackage{pgfplots}
\pgfplotsset{compat=1.17}
\usepackage[titlenumbered,ruled,linesnumbered]{algorithm2e}
\usepackage{algpseudocode}
\usepackage{mathtools,halloweenmath}
\usepackage{booktabs}
\newtheorem{theorem}{Theorem}[section]

\newtheorem{definition}{Definition}
\newtheorem{lemma}[theorem]{Lemma}

\newcounter{reviewer}
\newcounter{point}[reviewer]
\newcounter{question}[reviewer]
\renewcommand{\thepoint}{\,\thereviewer.\arabic{point}} 
\renewcommand{\thequestion}{\,\thereviewer.\arabic{question}} 

\newcommand{\shortreply}[2][]{\medskip \noindent \textbf{Reply}:\  #2
	\ifthenelse{\equal{#1}{}}{}{ \hfill \footnotesize (#1)}%
	\medskip}

\newenvironment{proofidea}{%
  \proof}{\endproof}
\usepackage{enumitem}
\usepackage{orcidlink}

\begin{document}

\title{MaskedHLS: Domain-Specific High-Level Synthesis of Masked Cryptographic Designs}

\author{\IEEEauthorblockN{Nilotpola Sarma\dag, Anuj Singh Thakur\dag, and  Chandan Karfa\dag}
\IEEEauthorblockA{\\\dag Indian Institute of Technology Guwahati\\
 \dag \{s.nilotpola, anuj.thakur, ckarfa\}@iitg.ac.in}}

\maketitle
\begin{abstract}
The design and synthesis of masked cryptographic hardware implementations that are secure against power side-channel attacks (PSCAs) in the presence of glitches is a challenging task. High-Level Synthesis (HLS) is a promising technique for generating masked hardware directly from masked software, offering opportunities for design space exploration. However, conventional HLS tools make modifications that alter the guarantee against PSCA security via masking, resulting in an insecure RTL. 
Moreover, existing HLS tools can't place registers at designated places and balance parallel paths in a cryptographic design which is needed to stop glitch propagation.  
This paper introduces a domain-specific HLS approach tailored to obtain a PSCA secure masked hardware implementation directly from a masked software implementation. 
It places the registers at specific locations required by the glitch-robust masking gadgets, resulting in a secure RTL. Moreover, our tool automatically balances parallel paths and facilitates a reduction in latency while preserving the PSCA security guaranteed by masking. Experimental results with the PRESENT Cipher's S-box \textcolor{black}{and AES Canright's S-box }masked with four state-of-the-art gadgets, show that MaskedHLS produces RTLs with \textcolor{black}{73.9\% decrease in registers and 45.7\% decrease in latency on an average} compared to manual register insertions. The PSCA security of the MaskedHLS generated RTLs is also shown with TVLA test. 

\end{abstract}

\section{Introduction}
\label{section:introduction}
\IEEEPARstart{E}{mbedded} devices implementing a cryptographic algorithm are susceptible to Power Side-Channel Attacks (PSCAs) \cite{kocherDPA}, where an attacker uses the target device's power consumption information to extract the secret values processed by the cryptographic algorithm. This is possible due to the existence of a direct correlation between the power consumption, which is a result of the overall transistor activity, and the computations being performed by the device under attack. Masking \cite{dom2016} is a countermeasure against such attacks. Masking splits the secret information processed in the algorithm into random shares drawn from a uniformly sampled random distribution. Thereafter, execution proceeds by processing these shares independently, taking care to \textit{re-randomize} computations that cause their recombination. This randomizes the results of intermediate computations, and hence the power consumption becomes random as well. Masking can be applied at the hardware \cite{dom2016, HPC2020, comarknichel2023} and software levels \cite{ProvablySMOAES2004, compilerassistedmasking2012}. 

Hardware masking must ensure resilience against the asynchronous behavior of circuits such as those caused by \textit{glitches} that may cause the recombination of shares within the circuit, removing the masking security. Secondly, there are hardware masking verification tools \cite{HPC2020} to verify that a handwritten masked hardware design is secure. However, they are limited in applicability due to gaps in hardware masking verification theory \cite{moos2019glitch}, which prevents the scalability of verification to higher orders. Further, keeping in mind the development of new masking schemes/gadgets, there is an increased need for design-space exploration at the hardware level. Thus, developing secure masked hardware from scratch requires significant expertise in the design, verification, and design-space exploration of masked designs. 

In contrast, software masking is easier to obtain from the algorithmic specification and easily verified \cite{compilerassistedmasking2012}. Therefore,  ways to obtain masked hardware from the corresponding masked software implementations would be beneficial. This is indeed a possibility, as the glitch-resistant hardware masking properties are a superset of the software masking properties. Also, most glitch-resistant hardware-masked gadgets like DOM \cite{dom2016}, HPC1 and HPC2 \cite{pini2020}, and COMAR \cite{comarknichel2023} have the same structure as their software-masked counterparts with additional registers to stop the propagation of glitches. Thus, in order to generate PSCA-secure masked hardware from masked software, a simple translation of intermediate code to RTL is desired. That can be followed by inserting registers at well-defined locations for state-of-the-art masking gadgets like DOM, HPC and COMAR. 

In this regard, High-Level Synthesis (HLS), which automatically generates register transfer level (RTL) hardware from descriptions in C/C++, can be helpful. 
A few recent works \cite{shortestpath2021} aim to utilize HLS to convert masked software implementations to masked hardware automatically. {\it In this work, we have shown that all stages of the HLS can hamper the security of masked circuits.} They have been discussed in greater detail in Subsection \ref{subsection:impactofHLS}. This suggests the need for a domain-specific HLS tool for masked hardware design focussing on the primary objective of keeping the side-channel security of the circuit intact throughout the HLS process. 

\textcolor{black}{We propose a domain-specific HLS tool called MaskedHLS, which performs a security-preserving translation of software-level cryptographic implementations into masked hardware. Specifically, the contributions of this work are as follows:} 
\begin{itemize}
    \item \textcolor{black}{We have analyzed the impact of HLS optimizations and the need for domain-specific HLS for PSCA-secure hardware design from masked software (in Section \ref{section:analysis}).}
    \item \textcolor{black}{We have utilized the concept of retiming to insert registers in designated locations and balance parallel paths with optimal latency and registers for gadget-based masked design to protect against glitches (in Section \ref{section:proposedflow}).}
     \item\textcolor{black}{The correctness of MaskedHLS is shown (in Section \ref{section:proofs}).}
     \item\textcolor{black}{A thorough experiment with PRESENT Cipher's S-box and the Canright's AES-256 S-box masked with DOM, HPC1, HPC2 and COMAR gadgets shows the usefulness of our approach. (in Section \ref{section:experiments}).} 
\end{itemize}

 \textcolor{black}{The MaskedHLS tool is generic enough to work for any masking gadget}. To the best of our knowledge, this is the first work that provides a complete HLS flow for PSCA secure hardware design from masked cryptographic code written in C/C++. 

 \textcolor{black}{The rest of the paper is organised as follows. The related works are discussed in Section \ref{section:relatedworks}. Section \ref{section:background} discussed the background material requried to understand the working of MaskedHLS. Following this, Section \ref{section:analysis} illustrates the impact of HLS on the PSCA security of masked designs and the motivation of our work. Section \ref{section:proposedflow} discusses the flow of the MaskedHLS tool and its steps in great detail. Section 
 \ref{section:proofs} discusses the correctness of our tool. Section \ref{section:experiments} discusses the results of using our tool on the selected benchmarks. Finally, Section \ref{section:conclusion} concludes the paper and discusses possible future directions of work.}

\section{Related Works}
\label{section:relatedworks}
Several works focusing on HLS of cryptographic implementations have been published \cite{hlsofchaskey2022, hlsofnttesl2020, zhang2019memory}. Works like \cite{konigsmarkhls2017,zhang2019memory} looked at the effects of various HLS optimizations on the side-channel security of unmasked code. These works were enough to establish that HLS does not possess security evaluation as a part of its design flow and generating secure hardware by HLS requires a case-by-case examination of all the optimizations, which is a very difficult task. However, these works do not consider masked cryptographic implementations and the effect of HLS on masking.

In 2021, Sadhukhan et al. \cite{shortestpath2021} demonstrated how to generate side-channel secure masked hardware in quick time using HLS. For their work, they used a 3-bit DOM-masked AES S-box with bit sliced implementation and generated the Verilog (RTL) for it using the \textit{Bambu} HLS tool \cite{bambu}.  They observed that HLS does not lead to side-channel secure hardware always. Hence, they examined the pragmas in the HLS software and came up with certain scenarios where an unguided application of pragmas would lead to side-channel leakage. They then presented remedies for better application of such pragmas. 
Recently, a study by Pundir et al. \cite{pundir2022analyzing}, highlights the importance of considering security when using HLS for hardware design. They have rightly pointed out that no HLS tool has been found to consider side-channel leakage while performing their code transformation procedures. There does not exist any work that develops domain-specific HLS tool for PSCA secure RTL generation from masked cryptographic design. 

\section{Background} 
\label{section:background}

\subsection{Glitch-Resistant masking}\label{subsection:glitchresistantmasking}
Hardware masking of cryptographic algorithms against PSCAs proceeds by breaking the inputs into independent random shares. For an affine component of the algorithm, these shares can be computed independently of each other to obtain the output shares. For an affine operation, $\oplus$ (XOR), such as in $c = a \oplus b$, can be broken into computations $c0 = a0 \oplus b0$ and $c1 = a1 \oplus b1$. Here $a$ and $b$ are broken into two shares initially as $(a0, ~a1): (a \oplus r1, ~r1)$ and $(b0,~ b1): (b \oplus r2, ~r2)$, where $r1$ and $r2$ are drawn independently from a uniform random distribution. Here $a0,~ b0,~ c0$ are $0$-shares and $a1,~ b1,~ c1$ are $1$-shares. Thereafter, $c0 \oplus c1$ gives the correct value of $c$. 
 
Similarly, an operation like $\otimes$ (bit-wise multiplication \footnote{\textcolor{black}{In this paper, $\otimes$ and $\&$ has been used interchangeably to mean bit-wise multiplication. $\oplus$ and $\hat{}$ has been used interchangeably to mean bit-wise XOR.}}), such as in $c = a \otimes b$, can be performed using shares $a0$, $a1$ and $b0$, $b1$. But to perform the multiplication operation, the four terms $a0 \otimes b0$, $a0 \otimes b1$, $a1 \otimes b0$, and $a1 \otimes b1$ must be calculated. Two of these computations, $a0 \otimes b1$, $a1 \otimes b0$, can not be performed without mixing the $0$-shares with the $1$-shares, which violates the independence of shares as assumed by masking theory. Hence, these operations need to be carefully remasked to ensure PSCA-security. 

Some algorithmic tricks can be used to mask these non-linear computations to optimize the amount of remasking. For example, the $SecMult$ algorithm by Rivain and Prouff \cite{rivain2010provably} proceeds by calculating the term $a0 \otimes b0$ separately and then performing masking with a random variable $r$  as  $(a0 \otimes b0) \oplus r$. The other terms are computed as $((((a1 \otimes b1) \oplus (a1 \otimes b0) \oplus (a0 \otimes b1)) \oplus r)$ following the parenthesized order. This requires two remasking operations, resulting in a masked hardware implementation for the multiplication operation. 

However, this multiplication algorithm does not remain secure in a glitchy circuit. Glitches are the phenomenon of different transition times in the signals of a circuit caused by variations in wire lengths and transistor speeds in circuits. As demonstrated in \cite{bilginthesis}, assuming that only one share $a1$ arrives later than the others, the number of times all the gates in the $SecMult$ circuit change values on different values of $b$ reveals a correlation between the power consumption and the value of $b$. Thus, masking in a glitchy circuit should be carefully handled. Several masking schemes were designed to be resistant to glitches in hardware \cite{ThresholdImplementations2006, prouff2011higherorderglitch, higherorderti2014, dom2016}.

One of the approaches towards glitch-resistant masking of cryptographic hardware is replacing all non-linear operations with \textit{gadgets} that are provably secure independently as well as in composition. A gadget is a probabilistic algorithm that takes n m-shares as inputs (where n is the number of inputs of the gadget) and returns a single m-shared output. {\it A gadget-based construction of masked circuits replaces one or more non-linear operations with secure gadgets.} Depending on the security guarantees provided by the gadgets, the glitch-robust security of the gadgets in composition can be guaranteed. In the following subsection, we briefly introduce those gadgets. 

\subsection{Multiplication Gadgets}
\label{subsection:gadgets}
Gross et al. \cite{dom2016} presented Domain-Oriented Masking (DOM) of hardware implementations of cryptographic algorithms against PSCAs 

where each input share corresponds to a domain. DOM  ensures that the computations corresponding to each share are carried out in their corresponding domain, and domains carry out computations independently of each other.
In this context, \textit{non-affine} operations require computations across domains, and these cross-domain computations require \textit{remasking} using new random values.
It was observed that glitches affected the combination of cross-domain shares, and hence, registers are used at those locations.
An example of a DOMAND gate (a multiplication gadget for one bit) is shown in Fig. \ref{fig:domandgadget}. 
Here, the products containing cross-domain terms,  $a0 \otimes b1$, $a1 \otimes b0$ are remasked using the same random value $r$ sampled from a uniform random distribution after which the outputs of the masking gates (XOR) are registered. 

\begin{figure*}[t!]
\begin{subfigure}[t]{0.20\textwidth}
    \centering
    \includegraphics[scale=0.26]{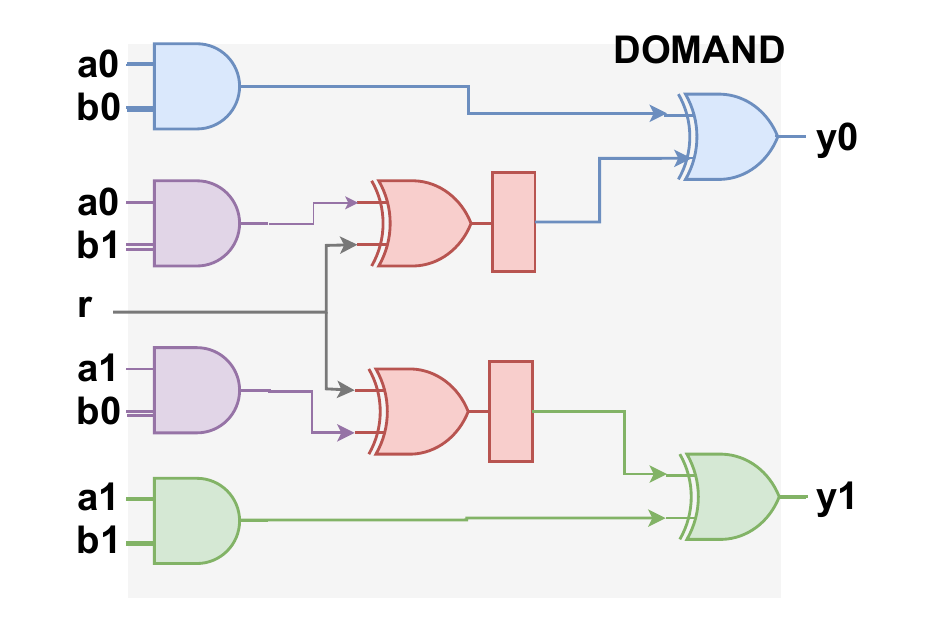}
    \caption{}
    \label{fig:domandgadget}
\end{subfigure}
   ~
\begin{subfigure}[t]{0.25\textwidth}
    \centering
    \includegraphics[scale=0.26]{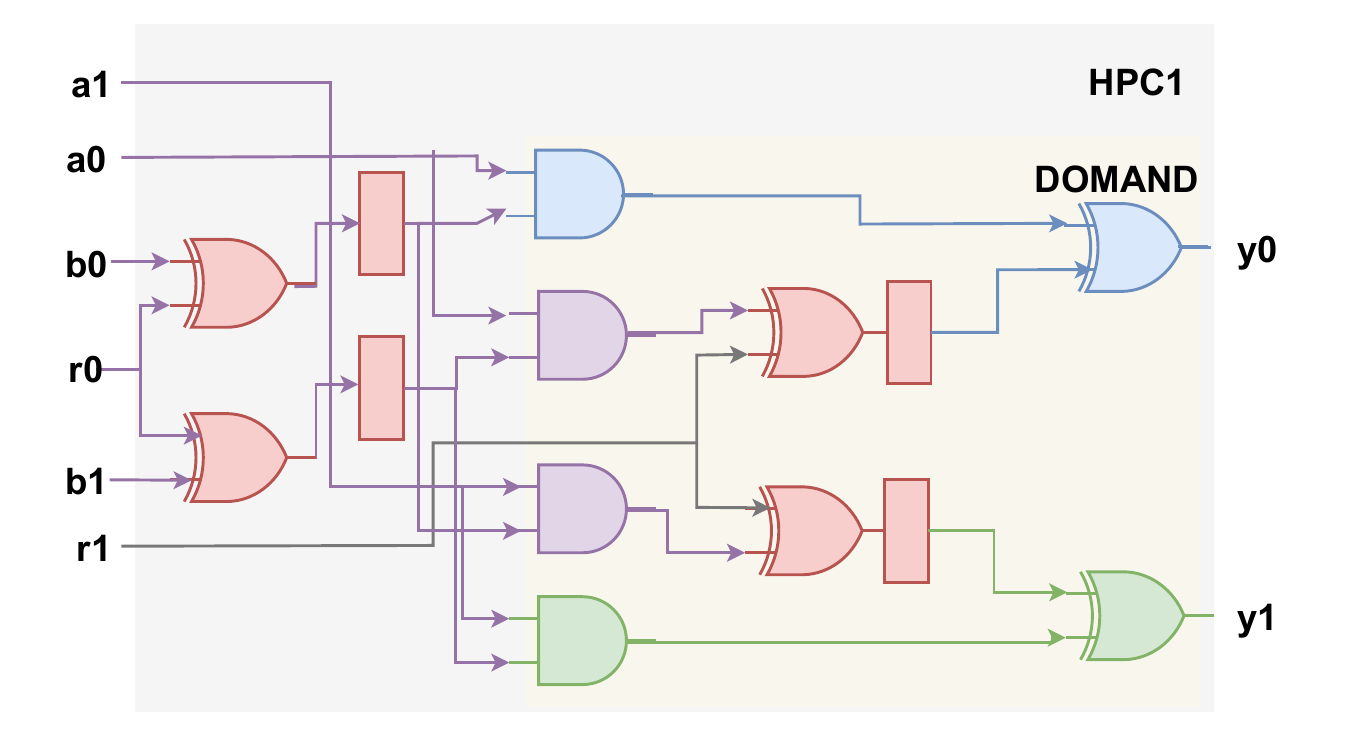}
    \caption{}
    \label{fig:hpc1gadget}
\end{subfigure}
\hfill
\begin{subfigure}[t]{0.22\textwidth}
    \centering
    \includegraphics[scale=0.13]{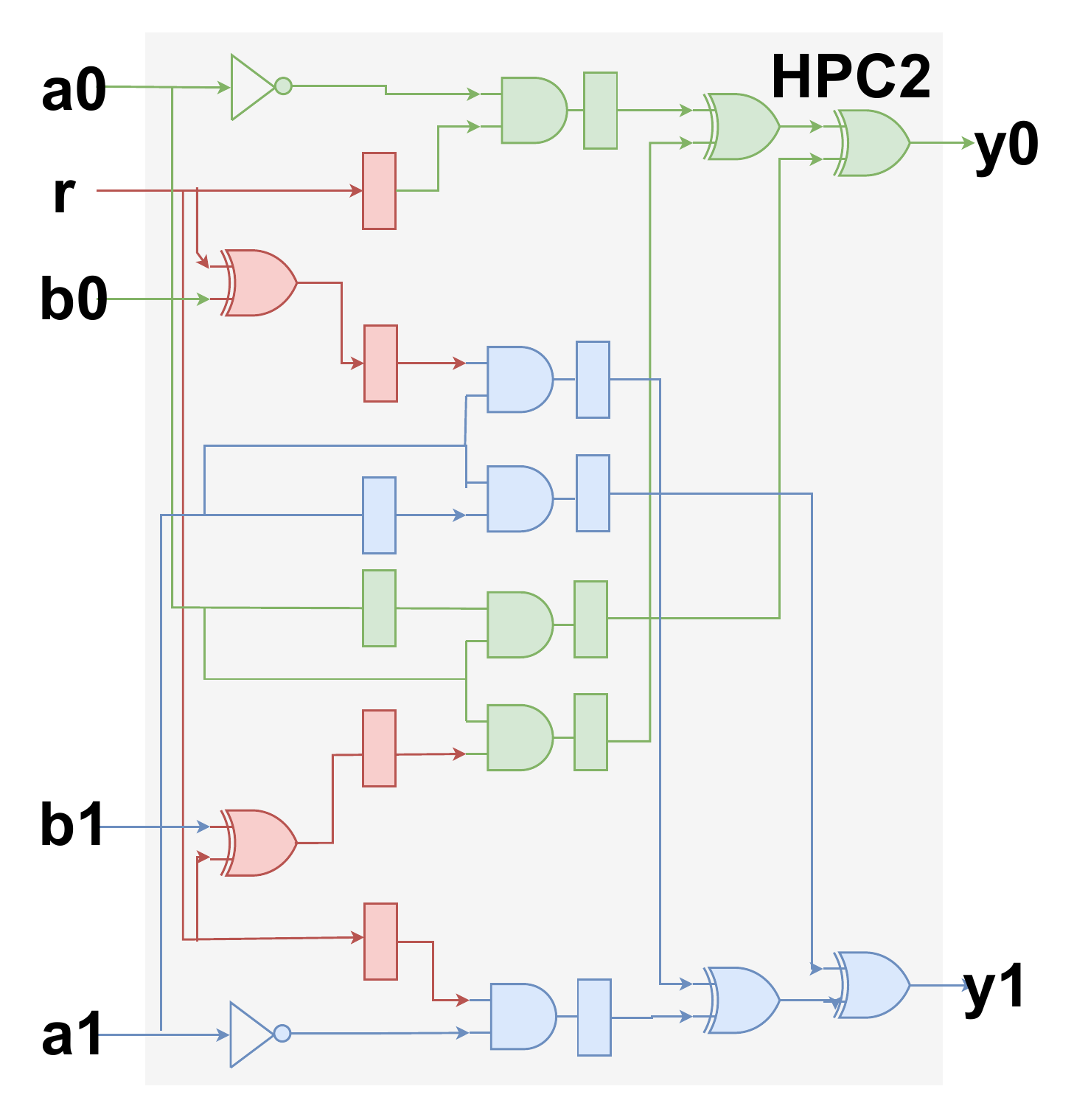}
    \caption{}
    \label{fig:hpc2gadget}
\end{subfigure}
\begin{subfigure}[t]{0.25\textwidth}
    \centering
    \includegraphics[scale=0.20]{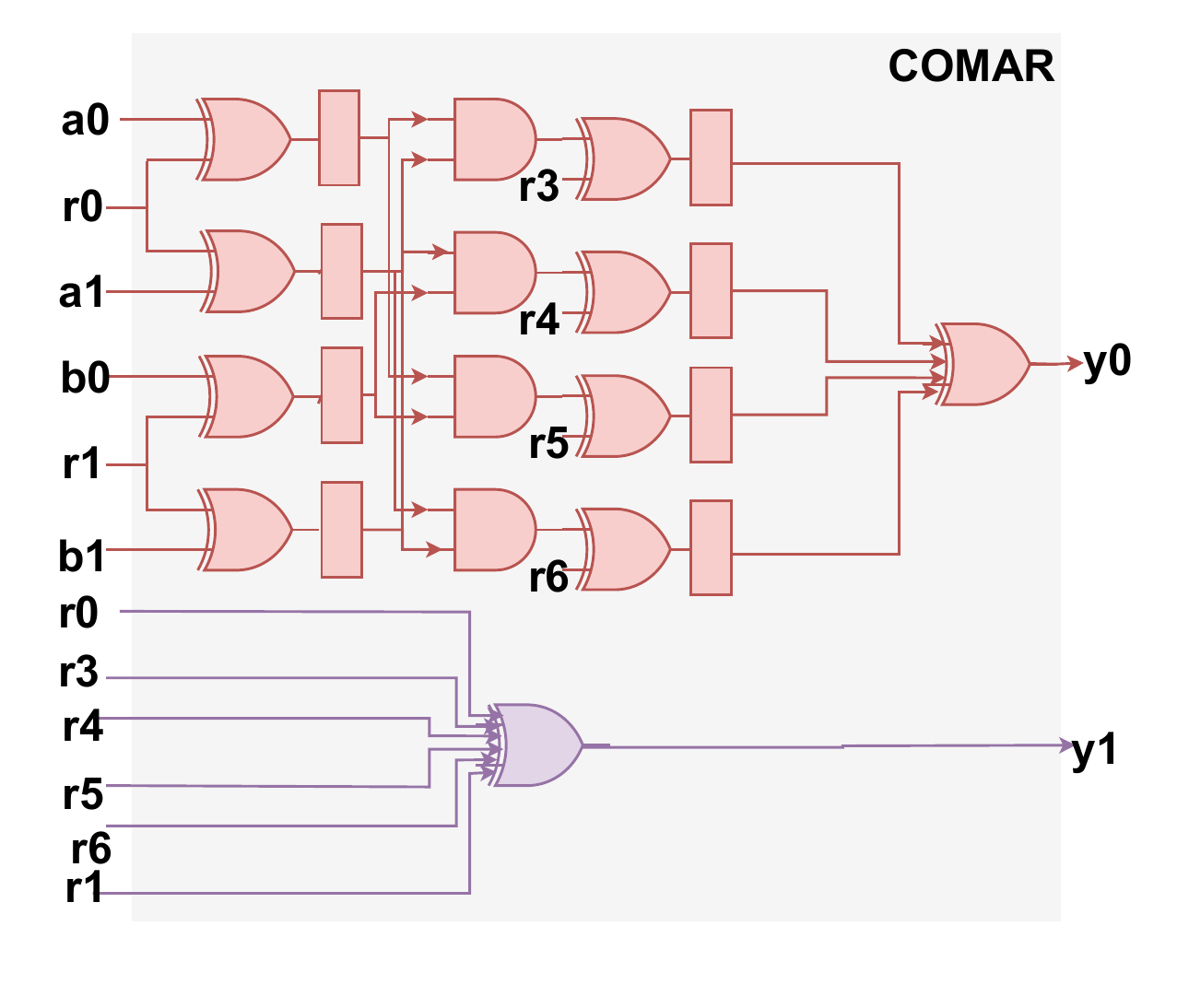}
    \caption{}
    \label{fig:comargadget}
\end{subfigure}
\caption{\textcolor{black}{Masked Multiplication Gadgets (a) DOMAND (b) HPC1 (c) HPC2 (d) COMAR.}}
\label{fig:gadgets}
\end{figure*}


A similar class of \textit{non-affine} gadgets were introduced in \cite{cassiers2020hardware}. The strategy was to refresh the masks of the gadgets before multiplying. The HPC1 gadget, proceeds by refreshing one of the inputs of the DOM gadgets (with two operands) using a refresh (remasking) gadget. For the computation $c = a \otimes b$, a HPC1 gadget masks both the shares of the input $b$ as : $(b0 \oplus r0)$ and $(b1 \oplus r0)$ and puts a register after these masked inputs before being input to the DOMAND circuit. The other inputs $a0$ and $a1$ are put into the DOMAND circuit. The HPC1 multiplication gadget is shown in Fig. \ref{fig:hpc1gadget}.

In HPC2 \cite{cassiers2020hardware}, all the inputs that have been split into shares of two are registered. Thus, one register each is placed after $a0$, $a1$, $b0$ and $b1$ for the computation $a \otimes  b$ in two shares. After that the computation is performed as follows: $c0 = ((a0 \otimes r) \oplus (b1 \otimes r)) \oplus (a0 \otimes b0)$ and $c1 = ((a1 \otimes r) \oplus (b0 \otimes r)) \oplus (a1 \otimes b1)$ with registers being placed at all the input shares and four intermediate locations. The HPC2 multiplication gadget using two shares is shown in Fig. \ref{fig:hpc2gadget}.   

Figure \ref{fig:comargadget} represents the COMAR gadget for $c = a \otimes b$. All the input signals are masked with the same mask bit $r$ for the $0$-shares and $r\prime$ for the $1$-shares. 
Four fresh mask bits r2 to r5 are used to mask the non-linear terms. As shown, the shared output is formed as $c0 = (((a0 \oplus r) \otimes (b0 \oplus r^{\prime})) \oplus r2) \oplus  (((a0 \oplus r) \otimes (b1 \oplus r^{\prime})) \oplus r3) \oplus  (((a1 \oplus r) \otimes (b0 \oplus r^{\prime})) \oplus r4) \oplus  (((a1 \oplus r) \otimes (b1 \oplus r^{\prime})) \oplus r5) $ and $c1 = r2 \oplus r3 \oplus r4 \oplus r5$.  
This gadget uses six masked bits which is larger than the HPC2 2-input AND gadget. However, all instantiated 2-input COMAR-AND gadgets in a circuit can use the same six random masks.

\subsection{\textcolor{black}{Retiming Basics}}
\label{subsection:retimingbasics}
Retiming \cite{parhi2007vlsi} is a widely used technique to change the locations of the registers in a design without affecting the input/output functionality of the design. In the following, we formalize the retiming process.

 A sequential circuit is represented by a directed graph $G(V, E)$ where each $\nu \in V$ is a design unit and each $e_{u,\nu} \in E$ is the edge corresponding to the flow of signal from the output of design unit $u$ to the input of design unit $\nu$ for any $u, \nu \in V$.
 Each edge $e_{u,\nu} \in E$ has an edge weight $w(e_{u,\nu})$ equal to the number of registers in that edge such that $w(e_{u,\nu}) \ge 0.$
 Each vertex $\nu \in V$ has a constant computational delay $d(\nu)$ such that $d(\nu)\ge0.$


\noindent
Given a  \textit{circuit} represented by a directed graph $G(V, E)$, a path $p$ is a sequence of alternating vertices and edges such that each edge is a fan-out of the previous vertex in the sequence such that: Computational Delay of the Path ($d(p)$) is the summation of the computational delays of all nodes in the path. Weight of the  Path $p$ ($w(p)$) is the summation of the weights of all edges $e \in E$ in this path.  A purely combinational path in a circuit will therefore have $w(p) = 0$. The clock period ($c$) of a circuit can thus be written as :

\begin{equation}
c = \mathop{max}_{\textbf{$p|w(p)=0$}}{d(p) }
\end{equation}

A  retiming label, $r(\nu)$, associated with each vertex $\nu \in V$ indicates the number of registers moved from the outputs to the input of the vertex $\nu$ associated with the retiming label. Retiming is defined as assigning retiming labels $r(u)$ to all the design units $u \in V$ of the circuit. If the edge weights for $e_{u,\nu} \in E$ in the original circuit, $G$, changes to an edge weight $w_r(e_{u,\nu})$ after retiming, then: 
        \begin{equation} 
            \textbf{$w_r(e_{u,\nu})$ = $r(\nu) + w_(e_{u,\nu}) - r(u)$ }
        \end{equation}

    Given a target clock period $c$, the minimum period global retiming of a circuit produces a retimed circuit subject to the following constraints on retiming labels:
 \begin{itemize}
     \item \textit{Feasibility Constraint (FC)}: For each edge $e_{u, \nu} \in E$, the edge weight $w_r(e_{u, \nu})$ in the retimed circuit must be non-negative, i.e., $w_r(e_{u, \nu}) \ge 0,  \forall   e_{u, \nu} \in E$. Using (2),
         \begin{equation}
             \textbf{$r(\nu) - r(u) \le  w_(e_{u,\nu})  , \forall   e_{u, \nu} \in E$}
         \end{equation}
     \item \textit{Critical Path Constraint(CPC)}: The delay $d(p)$ of all paths $p$ with $w(p)=0$ should be less or equal to the clock period after retiming. 
 \end{itemize}

 \noindent
 Consider any two nodes $u$ and $\nu$ in $G$. 
 There can be multiple paths from $u$ to $\nu$. 
 The minimum number of registers on any path from $u$ to $\nu$ is $W(u, \nu)$. 
Let the computational delays of all $n$ paths from $u$ to $\nu$ having $W(u,\nu)$ registers be $d(p_1), d(p_2), ..., d(p_n)$. Then  $D(u, \nu)$ is:
        \begin{equation}
            D(u, \nu) = \mathop{max}_{\textbf{$i=1$}}^{n}{d(p_i)}
        \end{equation}

\textit{With $D(u, \nu) > c$ for all paths from $u$ to $\nu$, $r(\nu) - r(u) + w(u, \nu) \ge 1$ must hold to make the critical path's computational delay $\le c$. }
Formally, the \textit{CPC} can be re-stated as: For all paths from $u$ to $\nu$ with $D(u, \nu) > c$, 
     \begin{equation}
         r(u) - r(\nu) \le w(u, \nu) - 1
     \end{equation}

Thus the objective of retiming is to identify the \textit{retiming labels} $r$ for all vertices that satisfy the constraints in equations (3) and (5). These can be solved using all pairs' shortest path as described in Section \ref{section:proposedflow}.

\section{\textcolor{black}{Analysis of the impact of HLS on PSCA security}}
\label{section:analysis}

In this Section, we first explore the impact of HLS optimizations on the PSCA security of masked hardware implementations. Following that, we discuss the need for automated optimal register insertion during the translation from masked software to masked hardware. 

\subsection{Impact of HLS on Power Side-Channel Security}\label{subsection:impactofHLS}

HLS comes with various optimizations that help in obtaining an area/latency-optimized RTL from C/C++ code. 
Given a gadget-based masked software code (without register annotations) of a cryptographic algorithm, HLS can convert it to an RTL design. We observe that the optimizations performed by HLS impact the PSCA security of the masked designs under consideration. Using case studies of VivadoHLS \cite{vivado} and Bambu\cite{bambu}, we present a few instances where HLS hampers the PSCA security of the masked hardware. 

\subsubsection{HLS front-end} The HLS front-end consists of the C compilation stage which translates the C/C++ code into an intermediate representation (IR) using a compiler like GCC  or LLVM. This phase applies optimizations like expression simplification, code motion, reassociation, etc. that may hamper the security guarantees made at C-level masked implementation. Below we present a few instances of such optimizations and illustrate how they hamper the side-channel security of the IR. 
 
{\it Reassociation:} LLVM compiler reassociates some of the intermediate computations causing incorrect recombination of shares within the algorithm. The fact that Bambu HLS does this was identified in \cite{shortestpath2021}.  For the input code in Listing. \ref{lst:domand1} and its interpretation in Figure \ref{fig:domandwithoutregs}, the XOR gates $i1$ and $i2$ are required after the cross-domain products $p2$ and $p3$ as specified in the C code. However, LLVM shifts the XOR gates to mask the products $p1$ and $p4$ instead. The absence of these XOR gates masking the outputs of $p2$ and $p3$ results in an unmasked circuit. Specifically, the computation of y0 in Listing. \ref{lst:domand1} is carried out as 
    $y0 = ((a0  \otimes  b1) \oplus  z ) \oplus (a0  \otimes  b0)$, ensuring that cross-domain computations are masked before recombination. The reassociation causes the computation to be carried out as  $y0 = ((a0  \otimes  b0) \oplus z ) \oplus (a0  \otimes  b1)$ instead. We also could not stop this optimization by LLVM through the Bambu tool version 0.9.6 with $\#pragma~HLS~\_interface ~\langle variable \rangle~ none\_registered$ as done in \cite{shortestpath2021}.
    
{\it Expression Balancing in GCC:}  Many times C/C++ code is written with a sequence of operations, resulting in a long chain of operations in RTL after HLS. This can increase the delay in the design. By default, VivadoHLS rearranges the operations using associative and commutative properties. This rearranges operators to construct a balanced tree and reduces delay. This optimization might hamper the security of the masked circuit. In Listing \ref{lst:expressionbalancing} for example, we have the DOMAND software masked code which gets reassociated into the expression: $y0 = ((a0  \otimes  b0) \oplus  z ) \oplus (a0  \otimes  b1)$  as a result of these optimizations by GCC.

For integer operations expression balancing is enabled by default but can be disabled using the $\#pragma~HLS~EXPRESSION\_BALANCE~off$ directive as shown in Listing. \ref{lst:resourceallocation}. For floating-point operations, expression balancing is disabled by default but may be enabled using the $\#pragma~HLS~EXPRESSION\_BALANCE$.

Thus it is clear that the designer needs precise knowledge of all the optimizations to avoid such consequences.

\begin{lstlisting}[ 
    basicstyle=\footnotesize, caption={DOMAND ecpression},captionpos=b, label={lst:expressionbalancing}] 
 1. #include "ap_int.h"
 2. ap_int<9> domand (ap_int<9> a0, ap_int<9> a1, 
 3. ap_int<9> b0, ap_int<9> b1, ap_int<9> z, 
 4. ap_int<9> *y0, ap_int<9> *y1) {
 5.    *y0 = ((a0 & b1) ^ z ) ^ (a0 & b0);
 6.    *y1 = ((a1 & b0) ^ z ) ^ (a1 & b1);
 7. return 0;}
\end{lstlisting}
\begin{lstlisting}[
    basicstyle=\footnotesize, caption={Resource-shared DOMAND},captionpos=b,  label={lst:resourceallocation}] 
 1. #include "ap_int.h"
 2. ap_int<9> multiply (ap_int<9>a0, ap_int<9>a1) {
 3. #pragma HLS INLINE off
 4. 	return a0 & a1;
 5. }
 6. ap_int<9>domand (ap_int<9>a0, ap_int<9>a1, 
 7. ap_int<9>b0, ap_int<9>b1, ap_int<9>z,
 8. ap_int<9>*y0, ap_int<9>*y1) {
 9. #pragma HLS EXPRESSION_BALANCE off
10. #pragma HLS allocation instances=multiply limit=2 
11. function //above pragma enables resource sharing
12. *y0 = (multiply(a0, b1) ^ z ) ^ multiply(a0, b0);
13. *y1 = (multiply(a1, b0) ^ z ) ^ multiply(a1, b1);
14. return 0;}
 \end{lstlisting}

\subsubsection{HLS Backend - Scheduling and Resource Allocation}
\begin{figure*}
    \centering
 \includegraphics[scale=0.30]{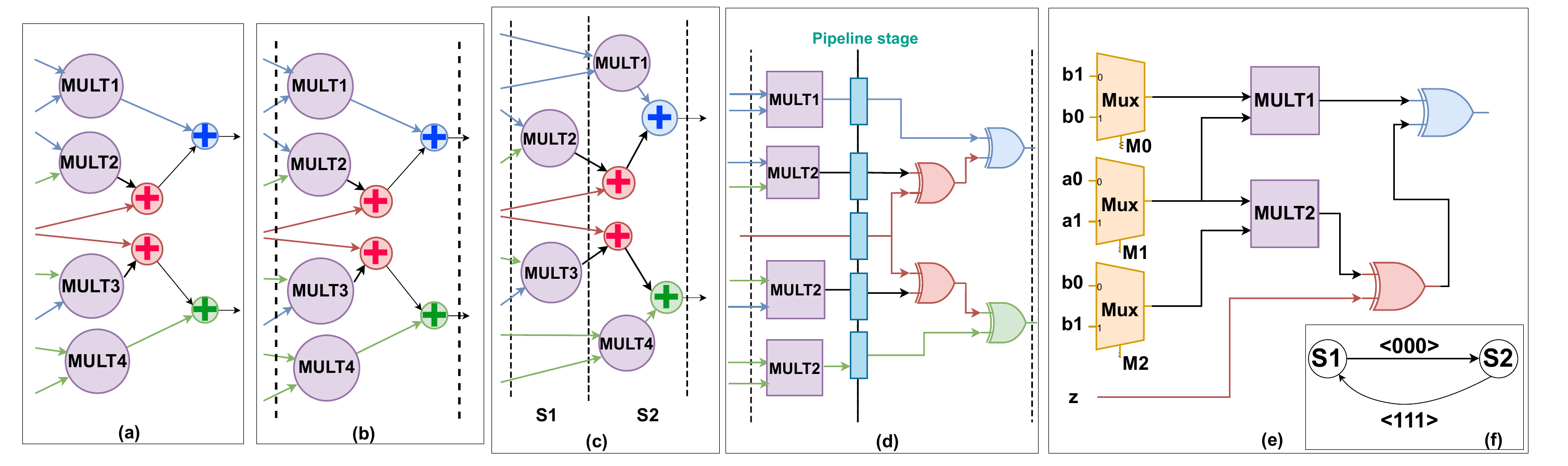}
    \caption{Example: (a) CDFG of the behavior in  Listing. \ref{lst:domand1}, (b) Schedule for 10ns, (c)  Schedule for 1ns , (d) Pipelined Design, (e) Resource-Shared Design (f) Controller for Resource-Shared Design}
    \label{fig:impactofhls}
\end{figure*}
After the preprocessing stage, based on the target clock period, the scheduler decides the number of time steps and the scheduled time of each operation. For example, if the target clock is 10ns for the example in Listing. \ref{lst:resourceallocation}, all the operations are scheduled in one clock by VivadoHLS as shown in Figure \ref{fig:impactofhls}b. A single-clock operation-chained datapath will be generated for the single-cycle schedule of Figure \ref{fig:impactofhls}b. It is pointed out in \cite{arxhls2023}, that such an operation chaining in the datapath also introduces side-channel vulnerabilities in the RTL. However, when the target clock is set to 1ns, the design is scheduled in 2 time steps as shown in Figure \ref{fig:impactofhls}c. The actual datapath depends on the resource optimization of the HLS tool. By default, most of the HLS tools generate a pipelined design as the one in Figure \ref{fig:impactofhls}d generated from the schedule in Figure \ref{fig:impactofhls}c. On the other hand, a user can specify resource constraints to restrict the area of the generated hardware. VivadoHLS allows the specification of resource bounds using pragmas like $\#pragma ~HLS~ resource\_allocation$ for a function or operation to restrict its number of instances. Consider the Listing. \ref{lst:resourceallocation} \footnote{Listing. \ref{lst:expressionbalancing}, \ref{lst:resourceallocation} and \ref{lst:domand1} are different representations of the same DOMAND behaviour. We took three variations to illustrate the various security implications of HLS.}. The number of multiplier instances has been restricted to 2 (line number 10). This results in a circuit with a datapath as shown in Figure \ref{fig:impactofhls}e where the resources are shared in a time-division multiplexed manner. 

To control the execution of the datapath, a controller FSM as shown in Figure \ref{fig:impactofhls}f will also be generated by the HLS tool. Here, in state S1, the controller will assign $\langle M0M1M2 \rangle = 000$ to execute the operations scheduled in state S1. Similarly, it will assign $\langle M0M1M2 \rangle = 111$ in S2 to execute the operations scheduled in S2. Such controller FSM may further introduce glitches in the datapath as shown in \cite{raghunathan1999register}. The PSCA security of the generated RTL may be compromised due to these glitches. Thus, additional analysis is needed for such a controller. 

\subsubsection{Discussion}

Thus, it is evident that masked designs are restrictive in terms of allowing for design-space exploration via rearrangement and resource-sharing. Additionally, the yet unexplored security vulnerabilities of the HLS optimizations on various other cryptographic implementations present a vast range of possible security vulnerabilities. However, the steps in converting masked software to masked hardware for state-of-the-art masking schemes like DOM \cite{dom2016}, HPC \cite{pini2020}, COMAR \cite{comarknichel2023} primarily require an operation by operation conversion into RTL from the IR. This should be followed by the insertion of registers at proper locations in the design to stop leakage due to glitches. For the DOMAND circuit in Figure \ref{fig:domandwithregs}, the registers $r01$ and $r10$ are required as shown in Section \ref{section:background}. Thus, it may not be advisable to use a generic HLS tool to directly generate PSCA secure masked hardware. Instead, the HLS process seeking to leverage the software-level masking security must focus on the optimal insertion of registers. In the next Subsection, we discuss how this can be optimally done.

\subsection{Motivation of our Work}\label{subsection:needforRegBal}

Given a software-level masked implementation of a cryptographic algorithm, we need to add registers in specific places in order to stop the propagation of glitches.

HLS tools have \textit{pragma} directives to allow such annotation.
\begin{figure}[t!]
\begin{subfigure}[t]{0.2\textwidth}
  \centering
  \includegraphics[width=1\linewidth]{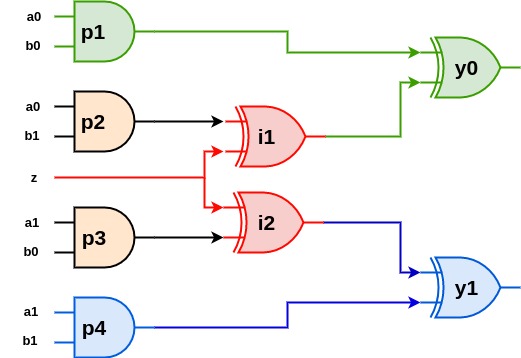}
  \caption{}
  \label{fig:domandwithoutregs}
\end{subfigure}%
\hfill
\centering

\begin{subfigure}[t]{0.2\textwidth}
  \centering
  \includegraphics[width=1\linewidth]{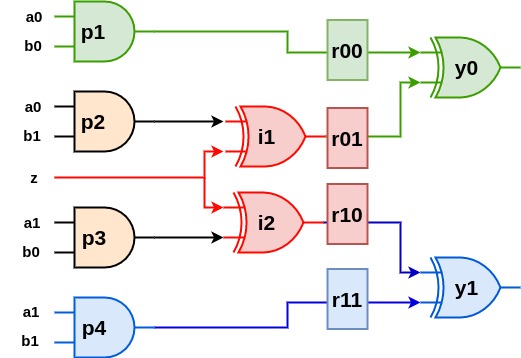}
  \caption{}
  \label{fig:domandwithregs}
\end{subfigure}
\label{fig:example1}
\caption{(a) Software-masked DOMAND hardware realization. (b) Hardware-masked DOMAND circuit with masking and balancing registers. }

\end{figure}
\begin{figure}
    \centering
    \includegraphics[scale=0.3]{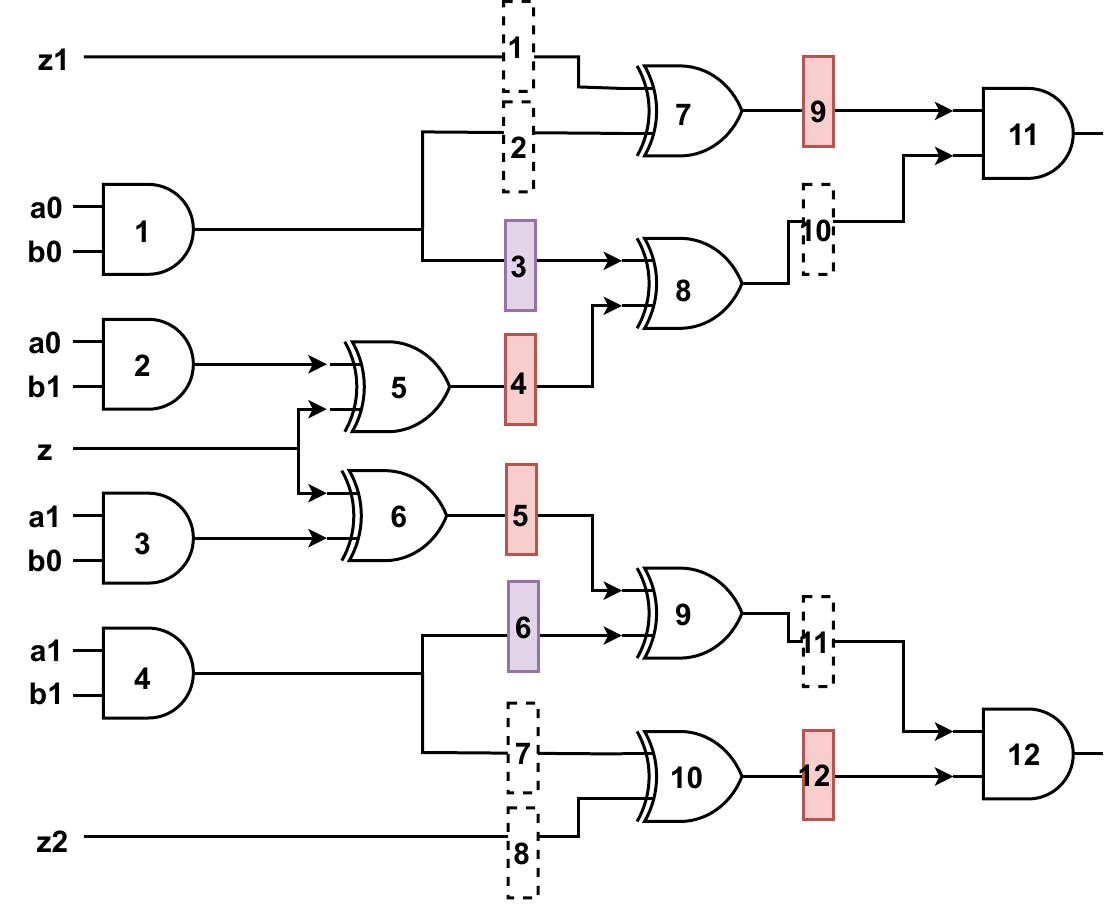}
    \caption{An example to illustrate the need for optimal register balancing in masked circuits. }
 
    \label{fig:motivation}
\end{figure}
However, there is no guarantee that HLS tool will enforce these pragmas due to the other constraints. For example, the Bambu HLS Version 0.9.6 ignored the $\#pragma~HLS~ none\_registered$ when applied on our example in Listing \ref{lst:domand1}. 
Further, a design may have many parallel paths. To preserve the latency of the circuit after the insertion of registers in specific paths, the parallel paths would also require register insertion. For example, consider the circuit given in Figure \ref{fig:domandwithoutregs}. This circuit corresponds to the DOMAND software specification in Listing \ref{lst:domand1} DOM-masked hardware requires the insertion of registers $r01$ and $r10$, as shown in Figure \ref{fig:domandwithregs}. With only these two registers, the inputs to the gates $y0$ / $y1$ have different latencies. This will results in incorrect circuit behavior. Thus, the \textit{balancing registers} $r00$ and $r11$ must be inserted as well. Figure \ref{fig:domandwithregs} is the circuit corresponding to an HLS-C input annotated as in Listing. \ref{lst:domand1}. For bigger circuits, there may be many parallel paths. Therefore register annotations that facilitate register insertion in parallel paths need automation.

For register insertion in multiple locations,  the number of \textit{balancing registers} and the design latency must be minimized as well. Consider the circuit in Figure \ref{fig:motivation}. Let us assume that the registers numbered 4, 5, 9 and 12 are required by a masking scheme for security. To insert registers 4 and 5, registers 1, 2, 3, 6, 7 and 8 need to be inserted to balance the parallel paths. Now, inserting registers 9 and 12 will require the insertion of registers 10 and 11 to balance the paths at gates $11$ and $12$. Thus a total of 12 registers need to be inserted. The overall latency of the circuit is now 2. However, careful examination of the circuit reveals that registers 1, 2, 7, 8, 10, and 11 can be optimized out. With the other 6 registers (3, 4, 5, 6, 9, and 12) the circuit has an overall latency of 1 and all the parallel paths in the circuit are balanced. This illustrates the need for optimal register balancing in masked circuits. 

\begin{figure}[htb]
\begin{minipage}[t]{0.5\linewidth}
\begin{lstlisting}[
    basicstyle=\footnotesize, caption={DOMAND C 
  code },captionpos=b, label={lst:domand1}
]
 1. int domand (bool a0, 
 2. bool a1, bool b0, 
 3. bool b1, bool r01, 
 4. bool *i1, bool *i2,
 5. bool z,
 6. bool *y0, bool *y1)     
 7. {p2 = a0 * b1;
 8. i1 = p2 ^ z;
 9. p3 = a1 * b0;
10. i2 = p3 ^ z;
11. p1 = a0 * b0;
12. p4 = a1 * b1;
13. *y0 = *i1 ^ p1;
14. *y1 = *i2 ^ p4;
15. return 0;}
     
\end{lstlisting}

\end{minipage}%
\begin{minipage}[t]{0.45\linewidth}
~~
\begin{lstlisting}[
    basicstyle=\footnotesize, caption={DOMAND with register annotations } ,captionpos=b, label={lst:domand2}
]
 1. int domand (bool a0, 
 2. bool a1, bool b0, 
 3. bool b1, bool z, 
 4. bool *i1, bool *i2, 
 5. bool *y0, bool *y1)    
 6. {p2 = a0 * b1;
 7.  i1 = reg(p2 ^ z);
 8.  p3 = a1 * b0;
 9.  i2 = reg(p3 ^ z);
10.  p1 = reg(a0 * b0);
11.  p4 = reg(a1 * b1);
12.  *y0 = *i1 ^ p1;
13.  *y1 = *i2 ^ p4;
14. } return 0;}
 \end{lstlisting}
\end{minipage}

\label{listing}
\end{figure}
 In our opinion, modern HLS tools perform too many optimizations which are counter-productive for PSCA secure hardware generation through HLS. There should be one-to-one translation from the C code to RTL. Moreover, register insertion and balancing are the most important measures to stop the propagation of glitches while maintaining minimum register usage and latency. None of the existing HLS tools can do these tasks in an automated way.  \textit{This calls for a domain-specific HLS tool for masked designs as proposed in this work}. Such a tool would not create vulnerabilities due to HLS and retain the security properties required while performing register balancing in an automated way. \textit{In this work, we have developed an automated register balancing at behavioural level using the concept of retiming \cite{parhi2007vlsi}.} \textcolor{black}{The basic concept of retiming is presented in Section \ref{subsection:retimingbasics}}.

\begin{figure*}[t!]
     \centering
     \begin{subfigure}[t]{0.22\textwidth}
         \centering
         \includegraphics[width=1\textwidth]{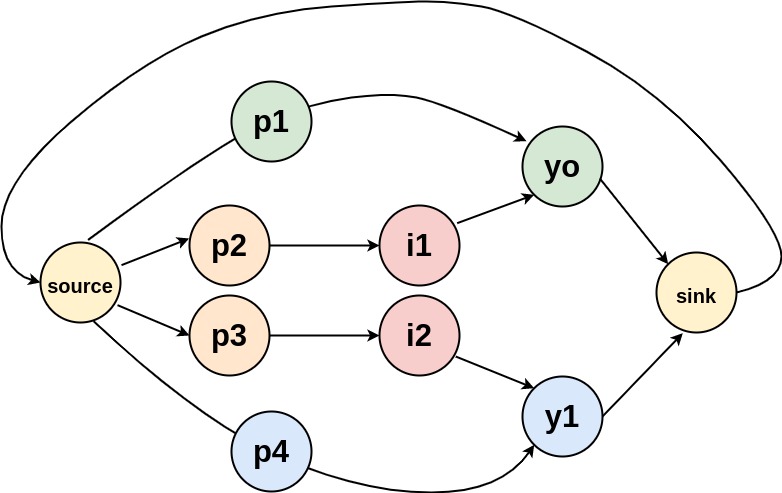}
         \caption{}
         \label{fig:ast_retiming}
     \end{subfigure}
     ~
     \begin{subfigure}[t]{0.25\textwidth}
         \centering
         \includegraphics[width=1\textwidth]{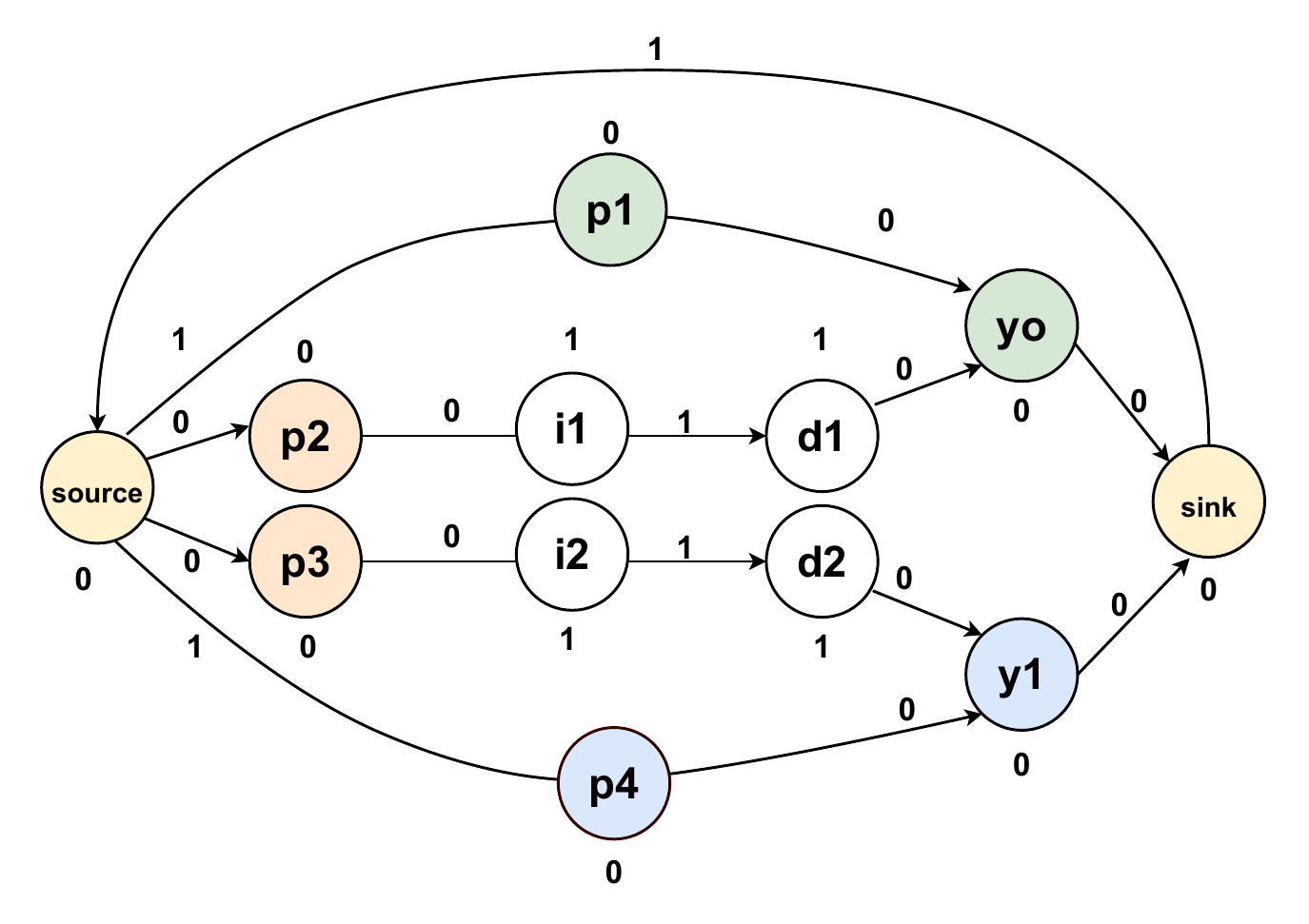}
         \caption{}
         \label{fig:HLS_model_back}
     \end{subfigure}
     \hfill
     \begin{subfigure}[t]{0.25\textwidth}
         \centering
         \includegraphics[width=1\textwidth]{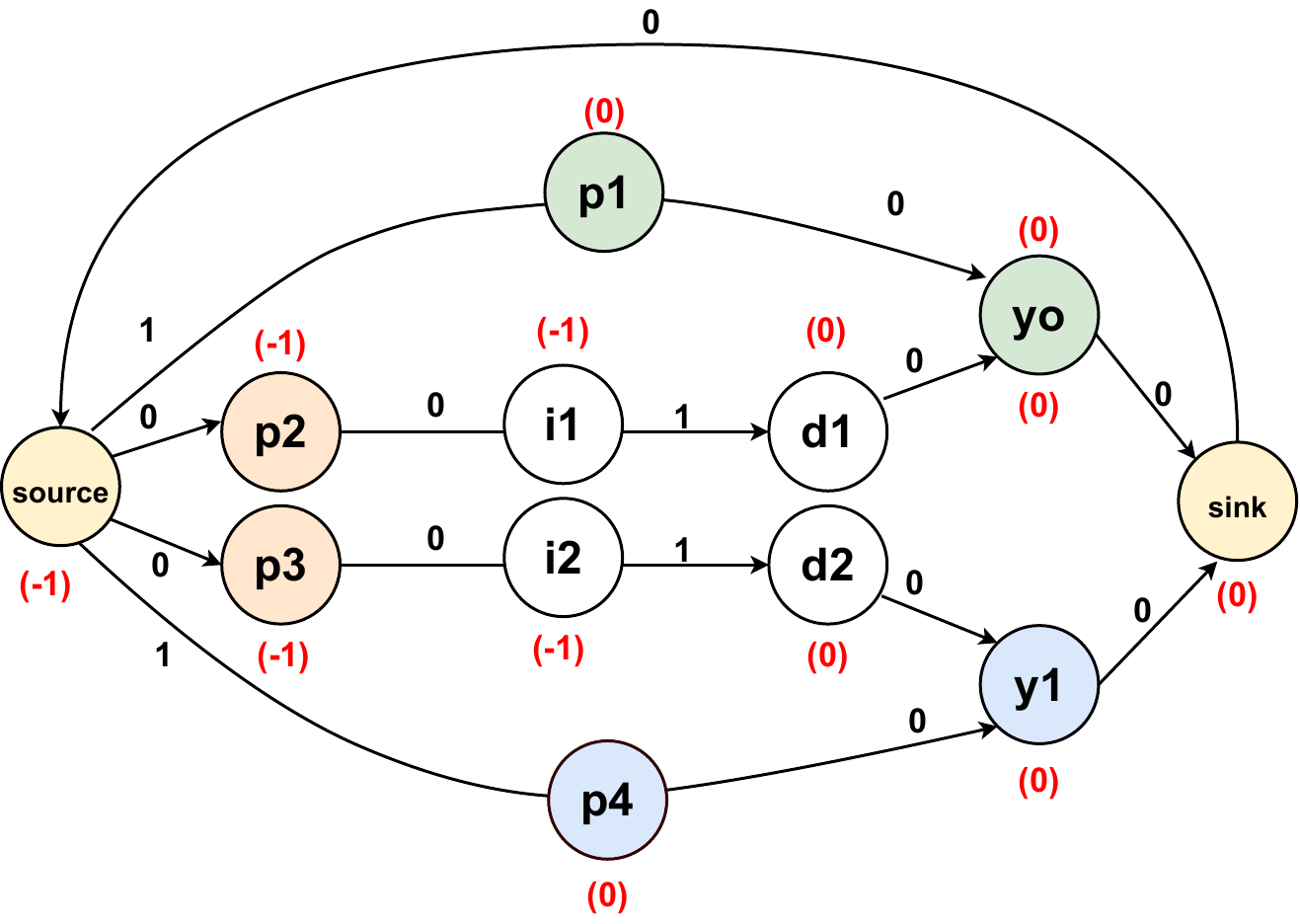}
         \caption{}
         \label{fig:HLS_model_final}
     \end{subfigure}
     ~ ~
     \begin{subfigure}[t]{0.22\textwidth}
         \centering
         \includegraphics[width=1\textwidth]{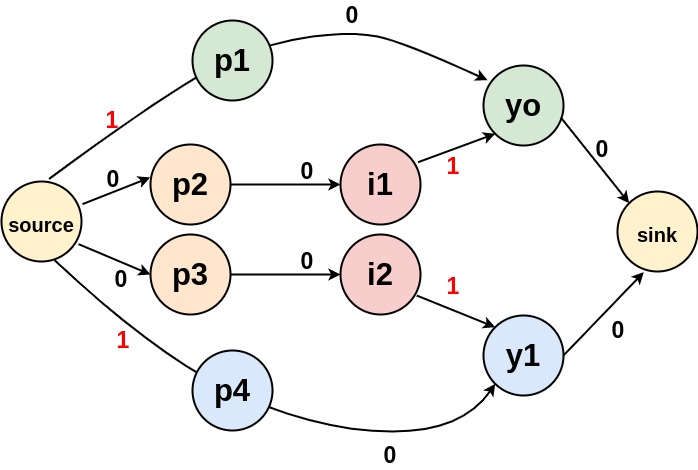}
         \caption{}
         \label{fig:ast_final_retimed}
     \end{subfigure}
    \caption{(a) AST before retiming, (b) \textcolor{black}{ HLS-model with back edge}, (c) HLS-model after retiming (retiming labels shown in parenthesis), (d) Final circuit after removing dummy nodes and back edge}
     \label{fig:retiming}
     \end{figure*}

\section{The Proposed MaskedHLS Flow}
\label{section:proposedflow}
A common approach towards masking cryptographic implementations is to replace the unmasked operations in the overall implementation with \textit{masked gadgets}. The state-of-the-art masked gadgets in DOMAND, HPC1, HPC2, and COMAR as discussed in Section \ref{subsection:gadgets} have locations for the insertion of registers. These masking gadgets at hardware and software differ only in the presence of registers for the hardware masking case. These registers ensure glitch-resistant masking security. Thus, in conversion from software to hardware-masked circuits we need to annotate the software-masked code with directives for the placement of registers. In VivadoHLS we could realize this using a template class $template <class T> T ~reg(T x)$ and later using it akin to a function call. In addition to inserting registers in specific locations, we also need to identify an optimal number of pipelined states and the registers need to be placed to balance parallel paths as discussed above. This can't be done by any existing HLS tools. 
 
Hence, we propose an HLS tool called MaskedHLS that makes this preservation of side-channel security the primary focus. 

\begin{figure}[htb]
    \centering
\includegraphics[scale=0.40]{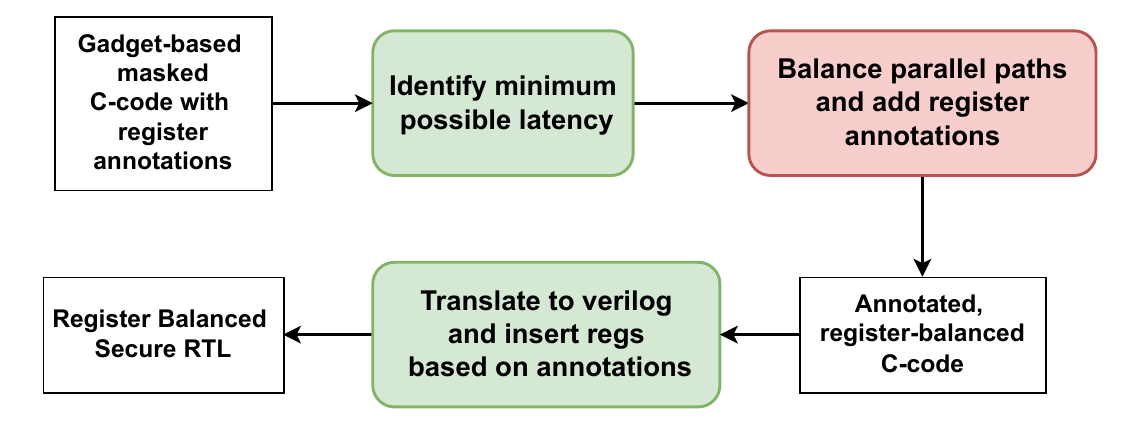}
    \caption{\textcolor{black}{Flow of MaskedHLS tool}}
    \label{fig:MaskHLSFlow}
\end{figure}

The input of MaskedHLS is a software implementation of a cryptographic algorithm that has been masked using gadgets. The gadgets have locations for register insertion for glitch-robust masking and these locations are indicated in the input software implementation using annotations. Given these inputs, MaskedHLS identifies the minimum possible pipeline stages in the circuit satisfying all necessary register requirements indicated by the annotations. In the next step, the register balancing procedure in MaskedHLS identifies all locations in parallel paths where registers need to be added. This produces an annotated C code. On this annotated C code, MaskedHLS performs a one-to-one translation into RTL code with registers inserted in all the places required by masking as well as balancing. Finally, MaskedHLS will create a pipelined RTL design. The overall flow of MaskedHLS is shown in Figure \ref{fig:MaskHLSFlow}. The steps are discussed in detail below. 

\subsection{Register Balancing at Behavioural Level} 
\label{subsection: regBalancing}
Given an unmasked software implementation in C/C++, the masked software is obtained by replacing the non-linear operations with the corresponding gadgets according to the masking scheme. The masking gadget/scheme specifies where registers must be inserted to maintain PSCA security in the corresponding hardware. These locations are indicated by annotations in the C/C++ input as $\langle lhs ~of~ operation \rangle$ = $reg(\langle rhs ~of~ operation \rangle)$. We need to put a register in those locations and balance parallel paths  automatically, with minimum pipelined stages. For this, the annotated C code is converted into an Abstract Syntax Tree (AST).

This AST has the same structure as the graph definition of the sequential circuits described in Section \ref{subsection:retimingbasics}. We develop a method that creates a special model of the AST and utilizes retiming logic on it to achieve the above goal. 

\subsection{Creating the HLS Model from the AST}
\label{subsection:creatingHLSModel}

Let us consider that the target clock period is $c$ in hardware implementation. For a given software code, the target clock period is always known. The AST is modified as follows to  create the \textit{HLS  model}:
\begin{itemize}
    \item \textit{Adding source and sink nodes: }A source node is added to the AST for all the inputs with edge weights 0, and a sink node is added for all the output nodes with edge weights 0.
    \item \textit{Adding a back-edge: }A directed edge is added from the sink node to the source node.
\end{itemize}
Such a model will allow us to add the additional registers in the back edge, and later, balancing will move them into the desired locations.
In the created \textit{HLS model}, we make the following changes to enable \textit{register balancing}:
\begin{itemize}
    \item \textit{Adding dummy nodes: }After each node $\nu \in V$, which has an annotation for a register insertion succeeding it, a \textit{dummy-node} $\nu^\prime$ is inserted. 
    \item \textit{Assigning computational delay to nodes: } The nodes $\nu$, after which registers must be added, and the dummy nodes $\nu^\prime$ are assigned computational delay of $d(\nu) = c$ and $d(\nu^\prime) = c$ respectively.
    All other nodes $u \in V$ apart from the ones assigned a computational delay of $c$ in the previous step are assigned computational delay $d(u) = 0$.
\end{itemize}  
  Any path $p$ involving the edge $e_{\nu,\nu^\prime}$, will have a delay of $2c$. Therefore, such a path will fail to meet the target clock period of $c$. Hence, a register must be inserted in that path at the location between $\nu$ and $\nu^\prime$ to meet the Critical Path Constraint (CPC) required for minimum period global Retiming as defined in Section. \ref{subsection:retimingbasics}.  By virtue of retiming, a register will be added in the parallel paths as well.

For the example in Listing \ref{lst:domand1}, a \textit{HLS model} is constructed in Figure \ref{fig:ast_retiming}.
One dummy node is inserted following each white-colored node (cross-domain nodes) and colored white as in Figure \ref{fig:HLS_model_back}. Assuming the target clock period is 1, all white-colored nodes are assigned the computational delay $d(\nu) = 1$. For all other nodes, the computational delay $d(u) = 0.$

\subsection{Finding the maximum number of register annotations in a path\label{SSS:maxReg}} 
Among all paths in the \textit{HLS model} between the source node and the sink node, the maximum number of annotations for register insertion is identified using a Depth First Search (DFS). We call it the \textit{maximum extra regs} in a path.
Once the \textit{HLS model} is converted to RTL by our HLS tool, we need to add these many registers in all parallel paths between source and sink. Therefore, this \textit{maximum extra regs} will determine the latency of the generated RTL. We will add those extra registers as weight in the back edge between the source and the sink. 
For all other edges, the edge weight is assigned to zero. It may be noted that the \textit{HLS model} is obtained from the software code, which initially had no register.

\begin{table}[]\caption{Feasibility constraints for the Circuit in Figure \ref{fig:HLS_model_back}}
    \centering
    \begin{tabular}{c c}
        $r(i2) - r(p3) <= 0$  & $r(i1) - r(p2) <= 0$\\
        $r(y1) - r(p4) <= 0$ & $r(y0) - r(p1) <= 0$\\
        $r(p4) - r(source) <= 0$ & $r(p3) - r(source) <= 0$\\
        $r(p2) - r(source) <= 0$ & $r(p1) - r(source) <= 0$\\
        $r(sink) - r(y1) <= 0$ &  $r(sink) - r(y0) <= 0$\\
        $r(source) - r(sink) <= 1$ & $r(d2) - r(i2) <= 0$\\
        $r(d1) - r(i1) <= 0$   
    \end{tabular}
    \label{tab:feasibilityconstraints}
\end{table}

    \begin{table}[]\caption{Critical path constraints for  Figure \ref{fig:HLS_model_back}}
        \centering
        \begin{tabular}{c c}
            $r(p4) -  r(d2) <= 0$ & $r(p4) -  r(d1) <= 0$\\
            $r(p3) - r(p4) <= 1$ & $r(p3) - r(p2) <= 1$\\
            $r(p3) - r(p1) <= 1$ & $r(p3) - r(i1) <= 1$\\
            $r(p3) - r(y1) <= -1$ & $r(p3) - r(y0) <= 1$\\
            $r(p3) - r(source) <= 1$ & $r(p3) - r(sink) <= -1$\\
            $r(p3) - r(d2) <= -1$ & $r(p3) - r(d1) <= 1$\\
            $r(p2) - r(p4) <= 1$
        \end{tabular}
        \label{tab:criticalpathconstraints}
    \end{table}
    
\subsection{Calculation of retiming constraints}
\label{subsection:calcretimingconst}

\textcolor{black}{For each node, we consider the retiming label $r(\nu)$. Feasibility constraints for each edge $e_{u,\nu}$ are calculated. For the HLS model in Fig. \ref{fig:HLS_model_back}, the Feasibility Constraints are shown in Table \ref{tab:feasibilityconstraints}. Critical Path constraints for each path from $u$ to $\nu$ such that $D(u, \nu) > 1$ computed. For the HLS model in Fig. \ref{fig:HLS_model_back}, some of the critical path constraints are shown in Table \ref{tab:criticalpathconstraints}.}

\begin{figure}
    \centering
  \includegraphics[scale=0.25]{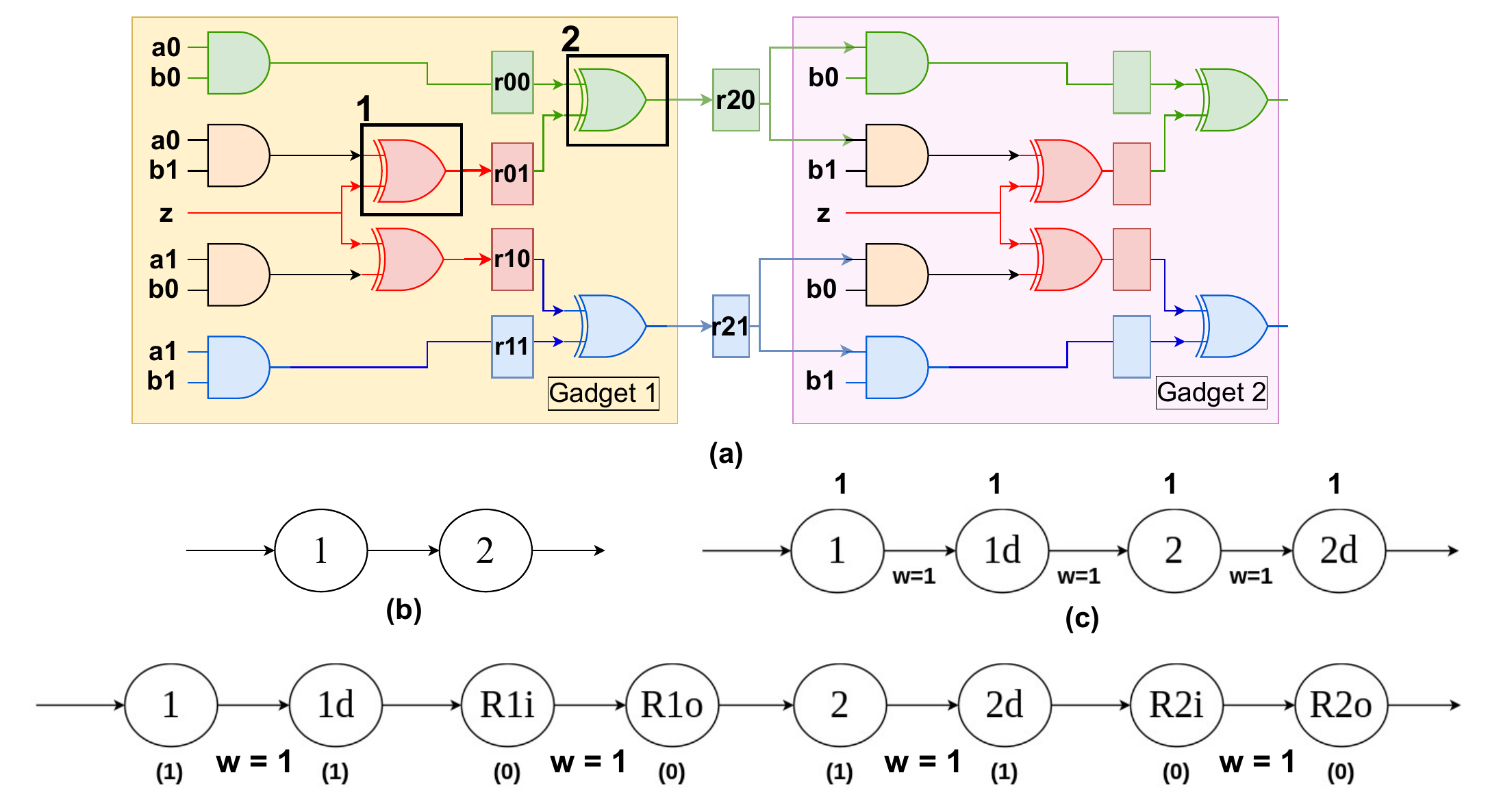}
    \caption{(a) The DOMAND-composed circuit. (b) Part of a circuit with two register insertions in series. (c) Register insertion in series. (d) Register insertion in series using an extra register.}
    \label{fig:series}
\end{figure}

\subsection{Inserting Registers in Series} 
\label{SS:reginSeries}
In a gadget based masking a situation may thus arise where in a single path there are two locations  where registers need to be inserted. Consider the two DOMAND gadgets composed with each other in Figure \ref{fig:series}a. To make the Gadget1-Gadget2 combination \textit{composable}, the registers $r20$ and $r21$ must be inserted as shown in Figure \ref{fig:series}b. The target clock is $c=1$. If we want to insert registers after gate $1$ and gate $2$ then we have to specify retiming constraints on both of them, thus, inserting a dummy node after each of them, as shown in Figure \ref{fig:series}c. Now the $d(\nu)$ values for the nodes become $d(1) = 1, d(1d) = 1, d(2) = 1, d(2d) = 1$. The $D(u, \nu)$ values are now: $D(1, 1d) = 2, D(1d, 2) = 2, D(2, 2d) = 2$. Thus, these $3$ edges $e_{1, 1d}$, $e_{1d, 2}$ and $e_{2, 2d}$  violate the \textit{Critical Path Constraints}. Hence $3$ registers are placed into the circuit at locations where $w=1$ in Figure \ref{fig:series}c. Here, the register between nodes $1d$ and $2$ is not needed, and as a result, the circuit is not balanced with minimum latency. The \textit{Adding Dummy Nodes} step from Subsection \ref{subsection:creatingHLSModel} is updated by the addition of these steps to address this issue: 
\begin{itemize}
    \item A redundant register at the edge between $1d$ and $2$ is deliberately inserted into the HLS model. This causes the critical path from $1d \rightarrow 2$ to break into two paths that meet the target clock $c=1$. This register is later removed after retiming.
    \item To ensure that this register is not moved by retiming, it is \textit{locked} with two nodes at its input and output, respectively. After the dummy node $1d$, two other nodes $R1i$ and $R1o$ are inserted. Similarly for node $2d$. 
    \item \textit{Register-lock} constraints are added for each $R$ . $r(RIn) == r(ROut)$. This constraint ensures that the number of registers moved into the edge $e_{RIn, ROut}$ equals the number of registers moved out of this edge. This means registers can move across this edge without affecting the existing edge weight; thus, it locks the register. 
\end{itemize}

This results in the AST in Figure \ref{fig:series}d and the subsequent steps can be performed on it.

\subsection{Finding the Retiming Labels}

\textcolor{black}{To find the values of \textit{retiming labels} that satisfy these constraints,  we construct a constraint graph  as follows: }
    \begin{enumerate}
        \item \textcolor{black}{For each \textit{retiming label} $r(\nu)$, a node $\nu ^\prime$  is created.}
        \item \textcolor{black}{If $N$ is the number of nodes in the circuit, a $N+1^{th}$ node is created.}
        \item \textcolor{black}{For each inequality $r(u) - r(v)\le k$, an edge $v^\prime \rightarrow u^\prime$ from node $v^\prime$ to node $u^\prime$ of weight $k$ is drawn. It is possible that $k < 0$ for some set of retiming labels as shown in the above example. }

        \item \textcolor{black}{For each node $v^\prime \in V^\prime$, an edge $N+1 \rightarrow v^\prime$ from the node $N+1$ to the node $v^\prime$ with weight 0 is drawn. At this point, the graph is guaranteed to not contain any negative edge cycle as shown in Lemma \ref{lemma:negativeweightcycle}.}
    \end{enumerate} 

Using Lemma \ref{lemma:shortestpath}, the shortest path from $N+1$ to any node  $v^\prime$ will give the correct \textit{retiming label} corresponding to $v^\prime$.
Since, there are no negative weight cycles in the constraint graph $G^\prime$ (as shown in  \ref{lemma:negativeweightcycle}), we can apply Dijkstra's single-source shortest path algorithm to obtain the  \textit{retiming labels} $r(v)$. The retiming labels obtained as a solution for the HLS model in Fig. \ref{fig:HLS_model_back} are shown (within parenthesis) for each vertex in Fig. \ref{fig:HLS_model_final}. The retiming labels satisfying all the constraints will give us the correct locations in the circuit where registers have to be inserted. The register balanced design obtained using these $r(v)$ values is shown in Fig. \ref{fig:ast_final_retimed}. After retiming, all dummy nodes and edges are removed. MaskedHLS will generate a register balanced C code from this HLS model. The register balanced, annotated C code, corresponding to Listing \ref{lst:domand1} is shown in Listing \ref{lst:domand2}. It may be noted that registers are added in the designated locations and in all parallel paths.

\subsection{Generating Pipelined RTL Design}
The last phase of maskedHLs takes this register-annotated C-code obtained in the previous step and generates RTL from it. The translation of this C-code to Verilog is done via a one-to-one mapping from the AST of the C to the RTL. In addition, the tool places registers according to the annotations in the C-code. Our tool does not apply any optimization in this phase. This effectively creates a pipelined RTL design with the number of pipeline stages equal to the maximum number of registers that have been added in a path (as identified in Sub-Section \ref{SSS:maxReg}).

\section{Correctness of MaskedHLS}
\label{section:proofs}
Here, we prove the correctness of our register balancing approach in MaskedHLS. We also show that MaskedHLS add the minimum number of pipelined stages. 

\begin{lemma}\label{lemma:shortestpath}
The shortest path from $N+1$ to $v^\prime$ in the constraint graph will give the \textit{retiming label} satisfying the constraints.
\end{lemma}
\begin{proof} (By induction)
Base: There is a direct edge form $N+1$ to each vertex $v^\prime$ with weight $0$. If this edge is the shortest path from $N+1$ to $v^\prime$, then the retiming label  $r(v^\prime) = 0$. It means there will be no registers moved across $v^\prime$. 

Now, assume the shortest path to $v^\prime$ is through $u^\prime$, i.e., $N+1 \xrightarrow{0} u^\prime \xrightarrow{-k} v^\prime$ is the shortest path. The edge $e_{u^\prime, v^\prime}$ came from the retiming constraint $r(v^\prime) - r(u^\prime) \le -k$. The retiming label of $u^\prime$ must be $0$. So the value of the shortest path to $v^\prime$, i.e., $-k$ will satisfy the constraint. 

Inductive Step: Now assume we have another vertex $u^\prime$ in the constraint graph with a direct edge to $v^\prime$ with the edge weight $w_{u^\prime, v^\prime} = l$. Let the shortest path to $u^\prime$ of length $-m$ (from the base case $m \ge 0$) already exist and be equal to the value of the retiming label of $u^\prime$: $r(u^\prime) = -m$.
Therefore given this node $u^\prime$, the shortest path from $N+1$ to $v^\prime$ either passes through $u^\prime$ or does not. 

Case I: The shortest path from $N+1$ to $v^\prime$ is the path $N+1 \xrightarrow{0} v^\prime$. The direct edge from $N+1$ to $v^\prime$ is the shortest path. Therefore, $w(N+1 \xrightarrow{-m} u^\prime \xrightarrow{l} v^\prime) \ge w(N+1 \xrightarrow{0} v^\prime)$. $\implies - m + l \ge 0$.
Putting the value of $r(u^\prime) = - m$ in this equation we get: $ l + r(u^\prime) \ge 0$$ \implies 0 - r(u^\prime) \le l$$\implies r(v^\prime) - r(u^\prime) \le l$ which is the constraint on the vertex $v^\prime$. Hence, the constraint is satisfied in this case.  

Case II: The shortest path from $N+1$ to $v^\prime$ is $N+1 \xrightarrow{-m} u^\prime \xrightarrow{l} v^\prime$. So the shortest path's weight is $- m + l$. Here, $r(v^\prime) = - m + l$. Now, we have to show that the constraint on $v^\prime$, $r(v^\prime) - r(u^\prime) \le l$ is satisfied with the retiming constraints. 
Putting the value $r(v^\prime) = - m + l$ in the constraint's RHS we have 
$r(v^\prime) - r(u^\prime) \implies -m + l - (-m) \implies l$ which is $\le l$. Hence, the constraint for this case is satisfied.
\end{proof}

To find the correct set of retiming labels, we have to find a solution to the constraint graph using the shortest path. For the shortest path to give a solution which is the correct retiming labels, the graph should contain no negative weight cycles as otherwise no solution can be reached using the shortest-path approach. 
\begin{lemma}\label{lemma:negativeweightcycle}
The constraint graph contains no negative weight cycles.
\end{lemma}
\begin{proofidea}
We start by considering a hypothetical negative weight cycle $C$ in the constraint graph, the weight of which is: $w_C = \sum_{i} w_{i, i+1}$, where $w_{i, i+1}$ denotes the weight of edge $e_{i, i+1}$ in $C$.
    We observe that for any cycle $C$ in the \textit{HLS model}, there is one or more (due to critical path constraints there may be multiple edges in the constraint graph corresponding to one edge between two nodes  in the \textit{HLS model}) cycle in the constraint graph derived from that cycle. Also, since there are no loops in the input circuit, thus the only cycles in the \textit{HLS model} will contain the edge $e_{sink, src}$. Hence, for each cycle in the constraint graph, there exists an equivalent path in the HLS model from $src \rightarrow sink$. The weights along this path represent the number of registers moved across the vertices in the path, which is the total number of registers contained in the path.
    Since a circuit cannot have a path from input to output with a negative number of registers, the sum of weights along the path must be non-negative.
    By removing the edge $e_{sink, src}$ from cycle $C$, we obtain a path from source to sink in the HLS model. The sum of weights along this path must also be non-negative.
    However, the weight of cycle $C$ is negative. This leads to a contradiction.
    Thus, our initial assumption of the existence of a negative weight cycle in the constraint graph is false.
\end{proofidea}
At this point, it is noteworthy that our \textit{register balancing} procedure will always terminate with a solution. It will never be the case that an infeasible set of constraints is generated for which there is no solution possible. As discussed in Section \ref{SSS:maxReg}, we identify the maximum number of registers are needed in a path and add that as weight in the edge between the source and sink. These registers are sufficient to satisfy all constraints.  Below in Lemma \ref{lemma:3} we prove the termination of our procedure. 

\begin{lemma}\label{lemma:3}
Register Balancing will always terminate with a solution resulting in the same latency as the number of registers inserted into the back edge.
\end{lemma}
\begin{proof}
Let the circuit obtained via register balancing using our method be $C$. Let the \textit{maximum extra regs} value we have obtained after the DFS of the AST with the annotations be $m$. For a minimal latency circuit, we need to have a circuit with $m$ registers in all parallel paths. Say our circuit $C$ has $m + k$ registers in a parallel path after register-balancing. Then, our circuit $C$ will have an un-optimal latency. We have to prove that such a scenario will never be reached by our \textit{register balancing} procedure. So let us assume there is a path $p$ after retiming with a weight $w(p) = m+k$ for some $k>0$. Then for this path $src \rightarrow v_i \rightarrow v_{i+1} \rightarrow ...\rightarrow v_N \rightarrow sink$, following from the convention of retiming rules in Section \ref{subsection:retimingbasics} where weight of each edge before retiming is $w(e_{j, j+1})$ and after retiming are $w_r(e_{j, j+1})$, we have:

\begin{center}
    $w_r(e_{src, v_i}) = r(v_i) - r(src) + w (e_{src, v_i})$\\
     $w_r(e_{v_i, v_{i+1}}) = r(v_{i+1}) - r(v_i) + w (e_{v_i, v_{i+1}})$ \\
     $\ldots$\\
     $w_r(e_{v_N, sink}) = r(sink) - r(v_N) + w (e_{v_N, sink})$
\end{center}
Adding them all, we get the weight of the path $w(p)$ to be,
\begin{center}
    $w_r(e_{src, v_i}) + w_r(e_{v_i, v_{i+1}}) + \ldots + w_r(e_{v_N, sink}) = w(p) = r(sink) - r(src)$ \\
    $\implies r(sink) - r(src) = m+k$
\end{center}
Since we must move the registers from the $sink \rightarrow src$ edge into the path $p$ via retiming, therefore $r(src) = -m$. $r(sink) = 0$. Therefore, following from above, 
$r(sink) - r(src) = m+k \implies 0 - (-m) \neq m + k$.
Which is a contradiction. 
Thus, our initial assumption is wrong. 
Hence, retiming results in a circuit with an optimal latency.
\end{proof}

\begin{lemma}
Retiming does not change the PSCA security of the circuit.
\end{lemma}
\label{lemma:security}
\begin{proof} 
\textcolor{black}{
The retiming procedure only inserts registers at the locations annotated in the input C code and the locations requiring balancing. Introducing registers at locations other than the locations annotated (balancing registers) does not compromise the security. Since we lock all existing registers using register locking constraints, the retiming will not move any existing registers. Therefore, there is no removal, insertion or movement of any circuit components during register-balancing, that can impact the security guaranteed by masking. Hence, retiming does not impact the PSCA security of the circuit.}    
\end{proof}

\subsection{Complexity Analysis}
\label{section:complexity}
\textcolor{black}{Our MaskedHLS tool's complexity is upper bound by the complexity of the register-balancing procedure. The calculation of the $D$ and $W$ matrices together takes $O(n^{3})$ time where $n$ is the number of nodes in the AST of input. This is because they can be obtained using all pairs shortest-path. Following that, the feasibility constraints are obtained for each edge of the graph in $O(n^{2})$ (as the number of edges in a graph is $O(n^{2})$) and critical path constraints for each edge $e_{u,v}$ where $D(u, v) > c$ which is at most $O(n^{2})$. These constraints are then modelled using a constraint graph which is linear in the number of constraints which is $O(n^{2})$. These constraints are solved again using Dijkstra's algorithm on the constraint graph which takes $O(n^{3})$ (i.e., $V+E$, $|V| = n$) here $n$ is the number of nodes in the original HLS model ($n$). Therefore the time complexity of the balancing procedure is $O(n^{3})$ in the number of nodes in the retiming model $n$.}

\section{Experimental Results}
\label{section:experiments}
\subsection{Implementation and Benchmark Details}
\label{subsection:implementation}
The MaskedHLS tool makes use of PyCparser \cite{bendersky2012pycparser} to parse the abstract syntax tree of the input C-code on which the balancing procedure and the one-to-one transformation to RTL are performed.  
We have tested our MaskedHLS tool on four different variants of the  PRESENT Cipher's 4-bit S-box\cite{bogdanov2007present} and Canright's AES-256 S-box \cite{canright2005compactSBox} masked using four different gadgets: the DOMAND gadget, the HPC1 gadget, the HPC2 gadget, and the COMAR gadget. The codebase, examples and scripts of MaskedHLS and the flow are available in github\footnote{Will be made public upon acceptance}.

\begin{table}[]\caption{\textcolor{black}{Results for MaskedHLS}}
    \resizebox{\columnwidth}{!}{%
        \begin{tabular}{|c|ccc|cccc|c|c|c|c|}\hline
            Design  &\#ann\_regs& \#bal\_regs& \#total\_regs & \textcolor{black}{\#C} &\textcolor{black}{\#nodes} & \#RTL & Runtime (s)\\\hline
             PRESENT\_DOMAND &16 &36 &52 & \textcolor{black}{83} & \textcolor{black}{105} & 299 & 0.33 \\
             PRESENT\_HPC1 &32 &68 &100 & \textcolor{black}{84} & \textcolor{black}{169} &454 & 0.40 \\
             \textcolor{black}{PRESENT\_HPC2} &\textcolor{black}{48}&\textcolor{black}{82} &\textcolor{black}{130} & \textcolor{black}{91} &  \textcolor{black}{168} &\textcolor{black}{420} &  \textcolor{black}{0.33}\\
             PRESENT\_COMAR &56 &38 &94  & \textcolor{black}{94}  & \textcolor{black}{209} &515  &0.88 \\\hline
             \textcolor{black}{ AES\_DOMAND} &\textcolor{black}{72} &\textcolor{black}{999}&\textcolor{black}{1071}  &  \textcolor{black}{485} & \textcolor{black}{1308} &\textcolor{black}{5307} &\textcolor{black}{11.22} \\
             \textcolor{black}{ AES\_HPC1} &\textcolor{black}{216} &\textcolor{black}{1689}&\textcolor{black}{1905}  & \textcolor{black}{515} &  \textcolor{black}{1668} &\textcolor{black}{7707} &\textcolor{black}{25.77} \\
             \textcolor{black}{ AES\_HPC2} &\textcolor{black}{432} &\textcolor{black}{1587}&\textcolor{black}{2019}  & \textcolor{black}{481} &  \textcolor{black}{1884} &\textcolor{black}{7875} &\textcolor{black}{64.87} \\
            \textcolor{black}{ AES\_COMAR} &\textcolor{black}{468} &\textcolor{black}{2486}&\textcolor{black}{2954} & \textcolor{black}{495}  & \textcolor{black}{2322} &\textcolor{black}{10261} &\textcolor{black}{119.17} \\\hline
        \end{tabular}
    }

    \label{tab:synthesis}
\end{table}

\subsection{MaskedHLS Synthesis Results}
\label{subsection:synthesisresults}
Table \ref{tab:synthesis} presents the results of MaskedHLS on all the eight test-cases. The runtime of MaskedHLS is dependent on the number of nodes being processed. \textcolor{black}{Specially, MaskedHLS takes an average of 54 seconds on the AES S-box designs with an average of 1795 nodes; and an average of 0.48 Seconds on the PRESENT S-box designs with an average of 422 nodes. AES\_COMAR took a significantly longer time to synthesize using MaskedHLS due to the higher number of constraints generated during register balancing due to a higher number of critical paths in the design for AES\_COMAR compared to the other AES S-box designs. The major part of the time is taken in register balancing.}

In Table \ref{tab:synthesis}, the number of registers annotated initially for gadgets (\#ann\_regs) and the number of additional registers inserted by MaskedHLS for balancing (\#bal\_regs), the total number of registers (\#total\_regs) and the lines of code in input C (\#C) and  RTL (\#RTL) are also shown. As seen in Table \ref{tab:synthesis}, the runtime of MaskedHLS on a 6-core Intel i7-8700 CPU operating at 3.20GHz is less than 1 sec for all PRESENT S-boxes and less than two minutes for all AES S-boxes.
\begin{table}[]\caption{\textcolor{black}{Area and timing overhead comparison with designs without registers.}}
    \resizebox{\columnwidth}{!}{%
        \begin{tabular}{|c|c|c|c|c|c|}\hline
            Design & Area(wo\_reg) & Area(w\_reg) & 
Timing(wo\_reg)  & Timing(w\_reg)\\\hline
            \textcolor{black}{PRESENT\_unmasked} & \textcolor{black}{940.11} & \textcolor{black}{NA} & \textcolor{black}{1.11} & \textcolor{black}{NA}\\
            PRESENT\_DOMAND & 2639.22 & 5546.05 & 1.64 & 0.93 \\
             PRESENT\_HPC1 & 2614.83 & 8815.18 & 1.69 & 0.83  \\
             \textcolor{black}{PRESENT\_HPC2} & \textcolor{black}{3220.92} & \textcolor{black}{10788.05} & \textcolor{black}{1.85} & \textcolor{black}{0.94} \\
             PRESENT\_COMAR & 2892.56 & 8936.05 & 1.82 & 0.98 \\\hline
              \textcolor{black}{PRESENT\_average} & \textcolor{black}{} & \textcolor{black}{2.97x} & \textcolor{black}{} & \textcolor{black}{0.52x} \\\hline
             \textcolor{black}{AES\_unmasked} & \textcolor{black}{55728.13} & \textcolor{black}{NA} & \textcolor{black}{18.99} & \textcolor{black}{NA} \\
             \textcolor{black}{AES\_DOM} & \textcolor{black}{1002202.72} & \textcolor{black}{1841877.19} & \textcolor{black}{28.01} & \textcolor{black}{5.72} \\
             \textcolor{black}{AES\_HPC1} & \textcolor{black}{1004612.10} & \textcolor{black}{3136636.19} & \textcolor{black}{28.94} & \textcolor{black}{4.12} \\
             \textcolor{black}{AES\_HPC2} & \textcolor{black}{1028632.47} & \textcolor{black}{2305685.39} & \textcolor{black}{31.60} & \textcolor{black}{4.50} \\
             \textcolor{black}{AES\_COMAR} & \textcolor{black}{134810.27} & \textcolor{black}{2727581.43} & \textcolor{black}{31.12} & \textcolor{black}{2.44} \\\hline
             \textcolor{black}{AES\_average} & \textcolor{black}{} & \textcolor{black}{6.85x} & \textcolor{black}{} & \textcolor{black}{0.13x} \\\hline
        \end{tabular}
    }
    wo\_regs corresponds to the gadget based masked circuit without registers, w\_reg corresponds to the output of MaskedHLS(gadget based masked circuit with registers according to the masking scheme and balancing registers in parallel paths). Timing is in nanoseconds.
    \label{tab:overhead}
\end{table}

The generated RTLs from MaskedHLS were synthesized to netlists using Synopsys Design Compiler (DC) using the TSL18FS120 cell library from Tower Semiconductor Ltd. at 180nm technology node. \textcolor{black}{To ensure that the downstream synthesis tool does not impact the security of the generated RTL via optimizations, we added commands (like $set\_dont\_touch$) in the synthesis script}. To compare the area and latency overhead due to balancing, the gadget-based masked c-codes for all designs sans the registers were synthesized to RTL. These, too, were converted to netlist using the Synopsys DC with the same library. The area and latency data from the DC's synthesis report were obtained for both versions of the designs while constraining the circuit to use only and, xor and invert gates and registers wherever necessary. Table \ref{tab:overhead} shows the comparison of total area and timing for all the designs against the versions without registers. It may be observed that the area has increased by 2.97x and 6.85x on an average for PRESENT S-boxes and AES S-boxes respectively, after inserting the register. We  have also added area and timing results for the AES S-box and PRESENT S-box designs in their native form (without masking) to show the area overhead due to masking (first row in each set of results in Table \ref{tab:overhead}). The area overhead of Masking is 5.9x, 9.4x, 11.5x, 9.5x for PRESENT S-box masked using DOM, HPC1, HPC2 and COMAR, respectively. For the Canright's AES S-box masked using DOM, HPC1, HPC2 and COMAR, the area overhead due to masking is 33.1x, 56.3x, 41.3x, 48.9x respectively.
This increase in area is because of the additional registers added by HLS and the technology mapping for a pipelined design does not allow for much area optimization versus the combinatorial circuits of the designs without registers that get largely optimized. The clock period (in ns) for designs generated by MaskedHLS is less due to the pipeline stages added through registers. 

\begin{table}[htb]
    \centering
    \caption{\textcolor{black}{Comparison of Register and Latency savings using MaskedHLS and Manual Methods}}
    \resizebox{\columnwidth}{!}{%
    \begin{tabular}{ |c|ccc|ccc| }
        \hline
        \textbf{Design} & \multicolumn{3}{c|}{\textbf{Registers}} & \multicolumn{3}{c|}{\textbf{Latency}} \\
        & \textbf{MaskedHLS} & \textbf{Manually} & \textcolor{black}{\textbf{Saving(\%)}} & \textbf{MaskedHLS} & \textbf{Manually} & \textcolor{black}{\textbf{Saving(\%)}} \\
        \hline
        \midrule
        PRESENT\_DOMAND & 52 & 168 & \textcolor{black}{69.0} & 3 & 5 & \textcolor{black}{40}\\
        PRESENT\_HPC1 & 100 & 290 & \textcolor{black}{65.5} & 5 & 9 & \textcolor{black}{44.5 }\\
        \textcolor{black}{ PRESENT\_HPC2} &  \textcolor{black}{130} &  \textcolor{black}{398} & \textcolor{black}{67.3} &  \textcolor{black}{5} &  \textcolor{black}{10} & \textcolor{black}{50}\\
        PRESENT\_COMAR & 94 & 570 & \textcolor{black}{83.5} & 5 & 9 & \textcolor{black}{44.5}\\
         \textcolor{black}{AES\_DOMAND} & \textcolor{black}{1071} & \textcolor{black}{4752} & \textcolor{black}{77.4} & \textcolor{black}{5} & \textcolor{black}{9} & \textcolor{black}{44.5}\\
          \textcolor{black}{AES\_HPC1} & \textcolor{black}{1905} & \textcolor{black}{6578} & \textcolor{black}{71.0} & \textcolor{black}{7} & \textcolor{black}{13} & \textcolor{black}{46.1}\\
           \textcolor{black}{AES\_HPC2} & \textcolor{black}{2019} & \textcolor{black}{8901} & \textcolor{black}{77.3} & \textcolor{black}{7} & \textcolor{black}{13} & \textcolor{black}{46.1}\\
        \textcolor{black}{AES\_COMAR} & \textcolor{black}{2954} & \textcolor{black}{14987} & \textcolor{black}{80.2} & \textcolor{black}{13} & \textcolor{black}{25} & \textcolor{black}{50}\\\hline
        \textcolor{black}{Average} & \textcolor{black}{ } & \textcolor{black}{ } & \textcolor{black}{73.9} & \textcolor{black}{ } & \textcolor{black}{ } & \textcolor{black}{45.7}\\
        \hline
        \bottomrule
    \end{tabular}}
    \label{tab:combined}
\end{table}

\subsection{Register Balancing Results}
MaskedHLS optimizes balancing registers and hence leads to a decrease in the number of registers in the RTL versus the circuit derived via conventional methods as discussed in Section \ref{subsection:needforRegBal}. \textcolor{black}{As can be seen in Table \ref{tab:combined}, on an average over both PRESENT S-box and AES S-box designs combined, MaskedHLS results in an RTL with 73.9\% lesser number of registers and 45.7\% less latency while ensuring PSCA-security versus the conventional approach where registers are placed in all parallel paths manually without any optimization as proposed in this work.} This result affirms our objective of obtaining minimum latency and registers.

\begin{figure*}
\centering
\includegraphics[scale=0.25]{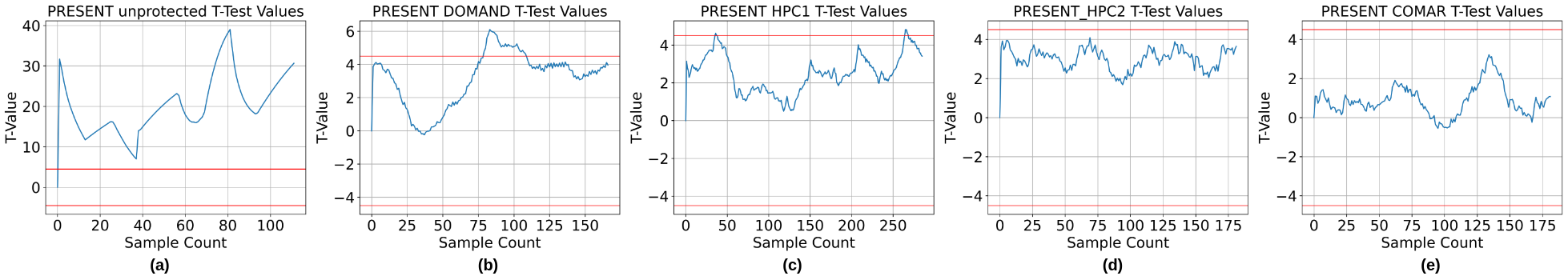}
    \caption{T-values for: (a) PRESENT\_unmasked. (b) PRESENT\_DOMAND. (c) PRESENT\_HPC1. \textcolor{black}{(d) PRESENT\_HPC2.} (e) PRESENT\_COMAR. (For each design, the x-axis contains T-values, y-axis contains the number of sample points per plaintext.)}
    \label{fig:tvlatestvalues}
\end{figure*}

\begin{figure*}
\centering
\includegraphics[scale=0.35]{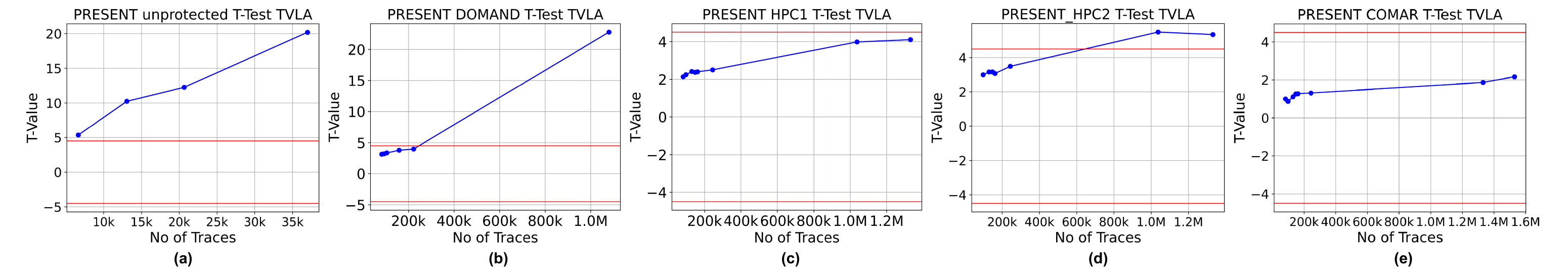}
\caption{TVLA values versus the number of traces: (a) PRESENT\_ unmasked. (b) PRESENT\_DOMAND. (c) PRESENT\_HPC1. \textcolor{black}{(d) PRESENT\_HPC2.} (d) PRESENT\_COMAR. (For each design, x axis contains T-values, y-axis contains the number of traces for which that TVLA value was observed.)}
\label{fig:tvla_results}
\end{figure*}

\begin{figure*}
\centering
\includegraphics[scale=0.27]{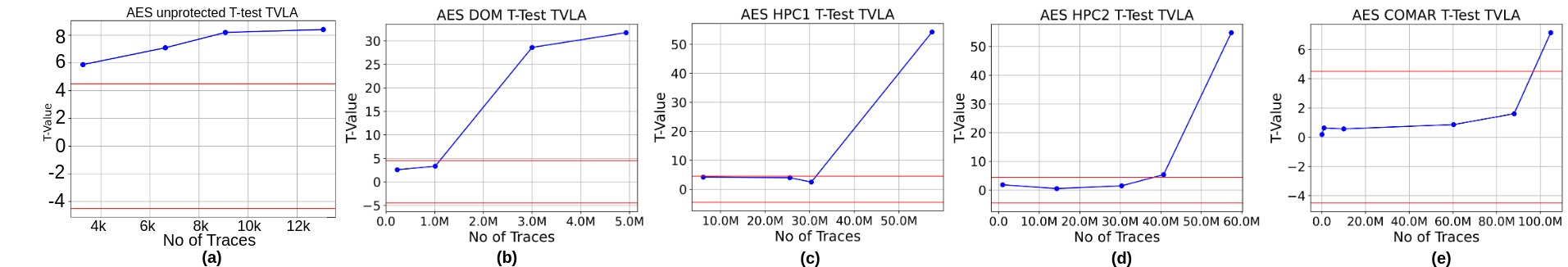}
    \caption{\textcolor{black}{TVLA values versus the number of traces: (a) AES\_unmasked. (b) AES\_DOMAND. (c) AES\_HPC1. (d) AES\_HPC2. (e) AES\_COMAR. (For each design, the x-axis contains T-values, y-axis contains the number of sample points per plaintext.)}}
    \label{fig:aestvla}
\end{figure*}


\begin{figure}[t!]
    \begin{subfigure}[t]{0.23\textwidth}
    \centering
  \includegraphics[scale=0.27]{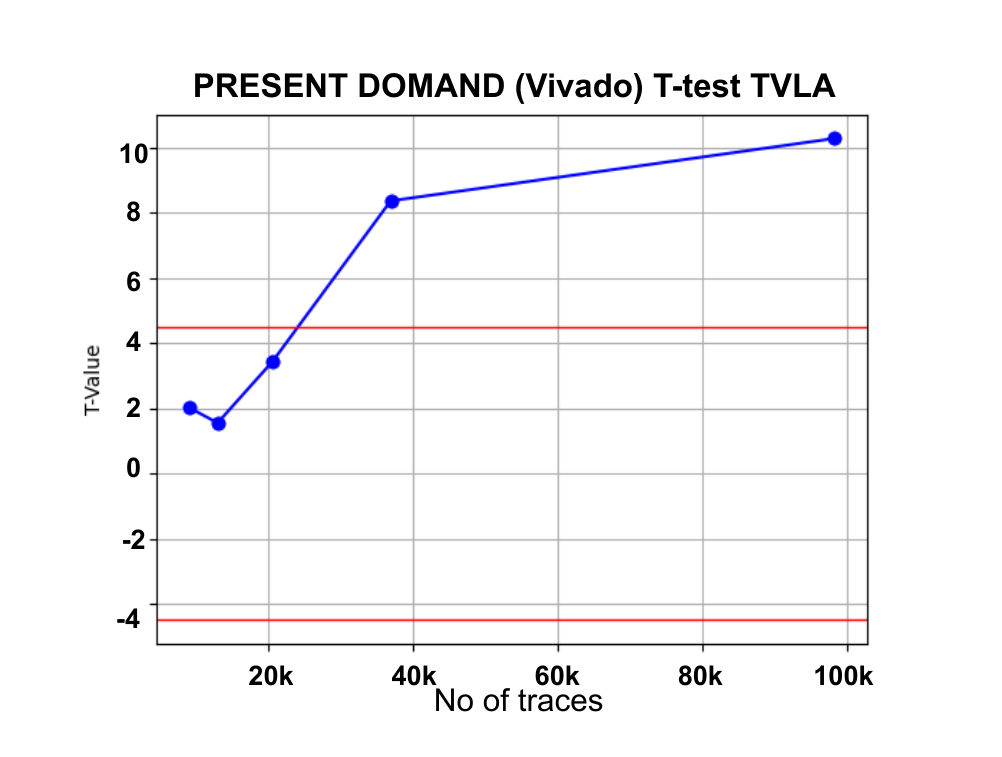}
    \caption{}
    \label{fig:vivadodomtvla}
\end{subfigure}
~~
\begin{subfigure}[t]{0.23\textwidth}
    \centering
  \includegraphics[scale=0.27]{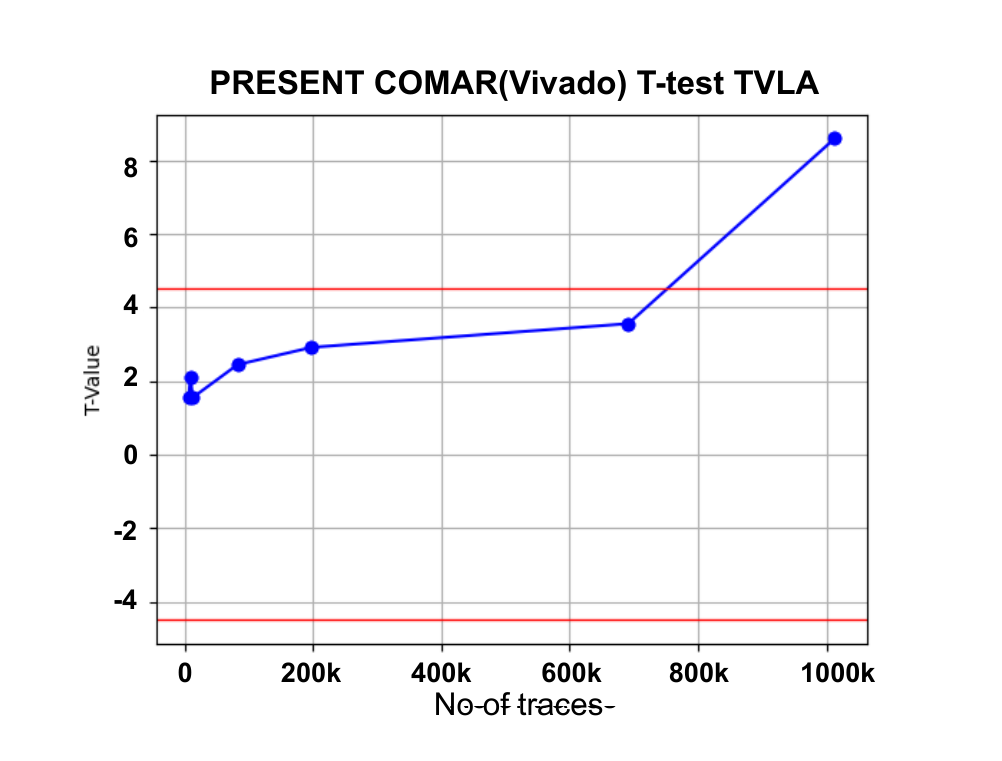}
    \caption{}
    \label{fig:vivadocomartvla}
\end{subfigure}
\caption{\textcolor{black}{TVLA values versus number of traces for: (a)  PRESENT\_DOMAND (b) PRESENT\_COMAR, both synthesized using Vivado HLS}}
\end{figure}

\subsection{PSCA Security Analysis}

It is necessary to verify that the output produced by MaskedHLS is indeed secure. We performed the Test-Vector Leakage Analysis (TVLA) \cite{tvla} of the power traces of  RTL obtained through MaskedHLS and compared them with those of unprotected (unmasked) design. Each RTL design was compiled into netlist using Synopsys DC and TSL18FS120 cell library. Then, the netlist was simulated using a testbench in the Synopsys VCS simulator. The switching activity of the circuit was dumped into the Value Change Dump (VCD) file. We then used Synopsys PrimeTime, which used netlist generated through DC and VCD file generated through VCS Compiler, giving the power traces in Fast Signal Database (FSDB) format. After that, Synopsys Custom Wave View tool was used to extract power traces in CSV format from the FSDB file. On this data, we used the conventional TVLA method \cite{tvla} to obtain the t-values.
The t-value corresponding to one plaintext for all PRESENT designs is shown in Figure\ref{fig:tvlatestvalues}. Clearly, the unprotected design is leaking. Among the PRESENT S-box designs, PRESENT\_HPC1, PRESENT\_COMAR  are more secure compared to  PRESENT\_DOMAND and PRESENT\_HPC2 whose t-value exceeded the Threshold of $\|4.5\|$ in  18\% and 16\% cases, respectively.

We extracted power traces, the number of which ranged between 5 thousand and 100 million. The objective was to check how good the protection was. The higher the number of traces where the t-value is not crossing the threshold of $\|4.5\|$, the more secure the design is. 
Figures \ref{fig:tvla_results}  shows the trend of TVLA-values for this experiment for the PRESENT designs.  The unprotected design crosses the $\|4.5\|$ mark for around 6 thousand traces. Whereas the COMAR and HPC1-masked designs are secure to 1.3 million traces.  The threshold value crosses the $\|4.5\|$ mark for around 220 thousand traces for DOMAND and around 600k for HPC2. HPC2 uses lesser random variables as compared to HPC1 as shown in Figure \ref{fig:hpc2gadget}. 

The design with COMAR is the most secure among all gadgets available. The experimental results are aligned with the Theoretical analysis of the gadgets. 

\textcolor{black}{We also performed TVLA test on the output of Vivado HLS \cite{vivado} on the DOMAND and COMAR masking gadget protected - PRESENT S-box. It can be seen in the results in Figure \ref{fig:vivadodomtvla} and \ref{fig:vivadocomartvla} that the security is significantly lesser (20k and 800k traces respectively) in terms of number of traces to obtain a correlation compared to MaskedHLS output ( $\ge 200$k traces and $\ge 1.3$ million traces respectively).  We observed similar results for PRESENT S-box using other gadgets (HPC1, HPC2). However, due to space limitations, we could not add all results. This reaffirms our motivation of domain specific MaskedHLS tool for PSCA-secure designs. Also, to test the efficacy of our tool on bigger benchmarks, we have tested our MaskedHLS tool with a Canright's AES S-box \cite{canright2005compactSBox} masked using the DOM, HPC1, HPC2 and COMAR gadgets. The TVLA results show that the key can not be revealed up to 1 million traces for AES\_DOMAND, 30 million traces for AES\_HPC1, 40 million traces for AES\_HPC2 and 100 million traces for COMAR as shown in Figure \ref{fig:aestvla}. The result for COMAR corresponds to the claims reported in the original proposal of the COMAR gadget \cite{comarknichel2023}.} Thus, our experiments clearly show that MaskedHLS generates PSCA-secure RTL from the masked software code. 

\section{Conclusion}
\label{section:conclusion}
Secure masked hardware design is a non-trivial task that requires significant time and expertise. Therefore obtaining masked hardware from masked software using HLS is beneficial. We have shown that existing HLS actually does not guarantee the PSCA security of the generated RTL. To address this shortcoming, we have developed MaskedHLS for generating PSCA secure RTL from the masked software version of the cryptographic designs. Experiments with two S-boxes for four gadgets show that MaskedHLS save on an average 73.9\% of registers and  45.7\% of latencies as compared to conventional processes. The TVLA analysis affirms the PSCA security of generated RTLs. The state-of-the-art PSCA-secure hardware design \cite{moos2019glitch} focuses on reducing the number of registers, design latency and randomness. In this regard, having minimum balancing registers is crucial. Our MaskedHLS generates RTL that uses minimum latency and registers to achieve PSCA security. In future, we plan to integrate randomness optimization strategies in our MaskedHLS tool. 

\bibliographystyle{IEEEtran}
\bibliography{reference}

\begin{thebibliography}{10}
\providecommand{\url}[1]{#1}
\csname url@samestyle\endcsname
\providecommand{\newblock}{\relax}
\providecommand{\bibinfo}[2]{#2}
\providecommand{\BIBentrySTDinterwordspacing}{\spaceskip=0pt\relax}
\providecommand{\BIBentryALTinterwordstretchfactor}{4}
\providecommand{\BIBentryALTinterwordspacing}{\spaceskip=\fontdimen2\font plus
\BIBentryALTinterwordstretchfactor\fontdimen3\font minus \fontdimen4\font\relax}
\providecommand{\BIBforeignlanguage}[2]{{%
\expandafter\ifx\csname l@#1\endcsname\relax
\typeout{** WARNING: IEEEtran.bst: No hyphenation pattern has been}%
\typeout{** loaded for the language `#1'. Using the pattern for}%
\typeout{** the default language instead.}%
\else
\language=\csname l@#1\endcsname
\fi
#2}}
\providecommand{\BIBdecl}{\relax}
\BIBdecl

\bibitem{kocherDPA}
P.~Kocher, J.~Jaffe, and B.~Jun, ``Differential power analysis,'' in \emph{AICC}.\hskip 1em plus 0.5em minus 0.4em\relax Springer, 1999, pp. 388--397.

\bibitem{dom2016}
H.~Gro{\ss}, S.~Mangard, and T.~Korak, ``Domain-oriented masking: Compact masked hardware implementations with arbitrary protection order,'' \emph{Cryptology ePrint Archive}, 2016.

\bibitem{HPC2020}
G.~Cassiers, B.~Gr{\'e}goire, I.~Levi, and F.-X. Standaert, ``Hardware private circuits: From trivial composition to full verification,'' \emph{IEEE TC}, vol.~70, no.~10, pp. 1677--1690, 2020.

\bibitem{comarknichel2023}
D.~Knichel and A.~Moradi, ``Composable gadgets with reused fresh masks $-$ first-order probing-secure hardware circuits with only 6 fresh masks,'' \emph{Cryptology ePrint Archive}, 2023.

\bibitem{ProvablySMOAES2004}
J.~Bl{\"o}mer, J.~Guajardo, and V.~Krummel, ``Provably secure masking of aes,'' in \emph{SAC}.\hskip 1em plus 0.5em minus 0.4em\relax Springer, 2004, pp. 69--83.

\bibitem{compilerassistedmasking2012}
A.~Moss, E.~Oswald, D.~Page, and M.~Tunstall, ``Compiler assisted masking,'' in \emph{CHES}.\hskip 1em plus 0.5em minus 0.4em\relax Springer, 2012, pp. 58--75.

\bibitem{moos2019glitch}
T.~Moos, A.~Moradi, T.~Schneider, and F.-X. Standaert, ``Glitch-resistant masking revisited: Or why proofs in the robust probing model are needed,'' \emph{IACR Transactions on CHES}, pp. 256--292, 2019.

\bibitem{pini2020}
G.~Cassiers and F.-X. Standaert, ``Trivially and efficiently composing masked gadgets with probe isolating non-interference,'' \emph{IEEE TIFS}, vol.~15, pp. 2542--2555, 2020.

\bibitem{shortestpath2021}
R.~Sadhukhan, S.~Saha, and D.~Mukhopadhyay, ``Shortest path to secured hardware: Domain oriented masking with high-level-synthesis,'' in \emph{ACM/IEEE DAC}, 2021, pp. 223--228.

\bibitem{hlsofchaskey2022}
S.~Inagaki, M.~Yang, Y.~Li, K.~Sakiyama, and Y.~Hara-Azumi, ``Examining vulnerability of hls-designed chaskey-12 circuits to power side-channel attacks,'' in \emph{2022 23rd ISQED}, 2022, pp. 1--1.

\bibitem{hlsofnttesl2020}
E.~Ozcan and A.~Aysu, ``High-level synthesis of number-theoretic transform: A case study for future cryptosystems,'' \emph{IEEE Embedded Systems Letters}, vol.~12, no.~4, pp. 133--136, 2020.

\bibitem{zhang2019memory}
L.~Zhang, D.~Mu, W.~Hu, Y.~Tai, J.~Blackstone, and R.~Kastner, ``Memory-based high-level synthesis optimizations security exploration on the power side-channel,'' \emph{IEEE Transactions on Computer-Aided Design of Integrated Circuits and Systems}, vol.~39, no.~10, pp. 2124--2137, 2019.

\bibitem{konigsmarkhls2017}
S.~C. Konigsmark, D.~Chen, and M.~D. Wong, ``High-level synthesis for side-channel defense,'' in \emph{IEEE 28th ASAP}.\hskip 1em plus 0.5em minus 0.4em\relax IEEE, 2017, pp. 37--44.

\bibitem{bambu}
C.~Pilato and F.~Ferrandi, ``Bambu: A modular framework for the high level synthesis of memory-intensive applications,'' in \emph{2013 FPL}, pp. 1--4.

\bibitem{pundir2022analyzing}
N.~Pundir, S.~Aftabjahani, R.~Cammarota, M.~Tehranipoor, and F.~Farahmandi, ``Analyzing security vulnerabilities induced by high-level synthesis,'' \emph{ACM JETC}, vol.~18, no.~3, pp. 1--22, 2022.

\bibitem{rivain2010provably}
M.~Rivain and E.~Prouff, ``Provably secure higher-order masking of aes,'' in \emph{CHES}.\hskip 1em plus 0.5em minus 0.4em\relax Springer, 2010, pp. 413--427.

\bibitem{bilginthesis}
B.~Bilgin, ``Threshold implementations: as countermeasure against higher-order differential power analysis,'' 2015.

\bibitem{ThresholdImplementations2006}
N.~Svetla \emph{et~al.}, ``Threshold implementations against side-channel attacks and glitches,'' in \emph{ICICS}.\hskip 1em plus 0.5em minus 0.4em\relax Springer, 2006, pp. 529--545.

\bibitem{prouff2011higherorderglitch}
E.~Prouff and T.~Roche, ``Higher-order glitches free implementation of the aes using secure multi-party computation protocols,'' in \emph{CHES}.\hskip 1em plus 0.5em minus 0.4em\relax Springer, 2011, pp. 63--78.

\bibitem{higherorderti2014}
B.~Bilgin, B.~Gierlichs, S.~Nikova, V.~Nikov, and V.~Rijmen, ``Higher-order threshold implementations,'' in \emph{ASIACRYPT}.\hskip 1em plus 0.5em minus 0.4em\relax Springer, 2014, pp. 326--343.

\bibitem{cassiers2020hardware}
G.~Cassiers, B.~Gr{\'e}goire, I.~Levi, and F.-X. Standaert, ``Hardware private circuits: From trivial composition to full verification,'' \emph{IEEE TC}, vol.~70, no.~10, pp. 1677--1690, 2020.

\bibitem{parhi2007vlsi}
K.~K. Parhi, \emph{VLSI digital signal processing systems: design and implementation}.\hskip 1em plus 0.5em minus 0.4em\relax John Wiley \& Sons, 2007.

\bibitem{vivado}
AMD, ``{Vivado HLS},'' \url{https://www.xilinx.com/products/design-tools/vivado.html}, 2022.

\bibitem{arxhls2023}
S.~Inagaki, M.~Yang, Y.~Li, K.~Sakiyama, and Y.~Hara-Azumi, ``Power side-channel attack resistant circuit designs of arx ciphers using high-level synthesis,'' \emph{ACM TECS}, vol.~22, no.~5, pp. 1--17, 2023.

\bibitem{raghunathan1999register}
A.~Raghunathan, S.~Dey, and N.~K. Jha, ``Register transfer level power optimization with emphasis on glitch analysis and reduction,'' \emph{IEEE TCAD}, vol.~18, no.~8, pp. 1114--1131, 1999.

\bibitem{bendersky2012pycparser}
E.~Bendersky, ``Pycparser c parser and ast generator written in python,'' 2012.

\bibitem{bogdanov2007present}
A.~Bogdanov, L.~R. Knudsen, G.~Leander, C.~Paar, A.~Poschmann, M.~J. Robshaw, Y.~Seurin, and C.~Vikkelsoe, ``Present: An ultra-lightweight block cipher,'' in \emph{CHES 2007}.\hskip 1em plus 0.5em minus 0.4em\relax Springer, 2007, pp. 450--466.

\bibitem{canright2005compactSBox}
D.~Canright, ``A very compact s-box for aes,'' in \emph{International Workshop on Cryptographic Hardware and Embedded Systems}.\hskip 1em plus 0.5em minus 0.4em\relax Springer, 2005, pp. 441--455.

\bibitem{tvla}
B.~G. Tobias \emph{et~al.}, ``Test vector leakage assessment (tvla) methodology in practice,'' in \emph{ICMC}, vol. 1001.\hskip 1em plus 0.5em minus 0.4em\relax sn, 2013, p.~13.

\end{thebibliography}

\vfill


\end{document}